\documentclass[lettersize,journal]{IEEEtran}
\usepackage{amsmath,amsfonts}
\usepackage{algorithmic}
\usepackage{algorithm}
\usepackage{array}
\usepackage[caption=false,font=normalsize,labelfont=sf,textfont=sf]{subfig}
\usepackage{textcomp}
\usepackage{siunitx}
\usepackage{stfloats}
\usepackage{url}
\usepackage{amssymb}
\usepackage{verbatim}
\usepackage{graphicx}
\usepackage{cite}
\usepackage{bm}
\usepackage{amsthm}
\theoremstyle{plain}
\usepackage{bbold}
\newtheorem{assumption}{Assumption}
\newtheorem{lemma}{Lemma}
\newtheorem{thm}{Theorem}

\newtheorem{coro}{Corollary}
\newtheorem{prop}{Proposition}

\begin{document}

\title{Over-the-Air Federated Multi-Task Learning via Model Sparsification and Turbo Compressed Sensing}

\author{Haoming~Ma,
	Xiaojun~Yuan,~\IEEEmembership{Senior Member,~IEEE,}
	Zhi~Ding,~\IEEEmembership{Fellow,~IEEE,}
	Dian~Fan,
	Jun~Fang,~\IEEEmembership{Senior Member,~IEEE}
	
	\thanks{H. Ma, X. Yuan, D. Fan and J. Fang are with the National Key Laboratory of Science and Technology on Communication, the University of Electronic Science and Technology of China, Chengdu, China (e-mail:hmma@std.uestc.edu.cn; xjyuan@uestc.edu.cn; df@std.uestc.edu.cn; junfang@uestc.edu.cn). Z. Ding is with the Department of Electrical and Computer Engineering, University of California at Davis, Davis, CA 95616 USA (e-mail:zding@ucdavis.edu). The corresponding author is Xiaojun Yuan.
	}
}



\maketitle

\begin{abstract} 
	To achieve communication-efficient federated multi-task learning (FMTL), we propose an over-the-air FMTL (OA-FMTL) framework, where multiple learning tasks deployed on edge devices share a non-orthogonal fading channel under the coordination of an edge server (ES). In OA-FMTL, the local updates of edge devices are sparsified, compressed, and then sent over the uplink channel in a superimposed fashion. The ES employs over-the-air computation in the presence of inter-task interference. More specifically, the model aggregations of all the tasks are reconstructed from the channel observations concurrently, based on a modified version of the turbo compressed sensing (Turbo-CS) algorithm (named as M-Turbo-CS). We analyze the performance of the proposed OA-FMTL framework together with the M-Turbo-CS algorithm. Furthermore, based on the analysis, we formulate a communication-learning optimization problem to improve the system performance by adjusting the power allocation among the tasks at the edge devices. Numerical simulations show that our proposed OA-FMTL effectively suppresses the inter-task interference, and achieves a learning performance comparable to its counterpart with orthogonal multi-task transmission. It is also shown that the proposed inter-task power allocation optimization algorithm substantially reduces the overall communication overhead by appropriately adjusting the power allocation among the tasks.
\end{abstract}

\begin{IEEEkeywords} 
	Federated multi-task learning, over-the-air computation, turbo compressed sensing.
\end{IEEEkeywords}

\section{Introduction}

\IEEEPARstart{W}{ith} the availability of a massive amount of data at mobile edge devices, there is a growing interest in providing artificial intelligence (AI) services, such as computer vision \cite{he_deep_2016} and natural language processing \cite{young_recent_2018}, at the edge of wireless networks. To utilize these data, conventional machine learning (ML) requires edge devices to upload their local data to a central node for model training. However, uploading such a huge volume of data by wireless communication incurs a huge cost of communication resources and compromises data privacy. To address these issues, federated learning (FL) has emerged as a popular framework for model training in a distributive and confidential manner \cite{goetz_active_2019}. In the FL framework, each edge device trains its local model based on its local data; and then transmits its local model parameters or gradients to an edge server (ES). The ES updates its global model parameters via model aggregation; and then broadcasts the updated global model to the edge devices. Compared with centralized learning, FL significantly relieves the communication burden and reduces the risk of data breaches.

In federated edge learning, a huge amount of model parameters need to be uploaded from massive distributed edge devices to the ES (referred to as uplink), where limited uplink channel resources (e.g., bandwidth, time, and space) become a critical bottleneck of efficient communication. To achieve communication-efficient FL, extensive research effort has been devoted to the uplink communication design \cite{kairouz_advances_2021,lin_deep_2020,wangni_gradient_2018,sattler_sparse_2019,konecny_federated_2017,zhu_broadband_2020,amiri_machine_2020,amiri_federated_2020}. The authors in \cite{lin_deep_2020,wangni_gradient_2018} pointed out that the local model parameters or gradients can be sparsified, compressed, and quantized before transmission to reduce the uplink communication cost without causing substantial losses in model accuracy. The authors in \cite{sattler_sparse_2019,konecny_federated_2017} proposed further improvements to reduce the number of communication rounds required for convergence. Different from digital model uploading in \cite{lin_deep_2020,wangni_gradient_2018,sattler_sparse_2019,konecny_federated_2017}, another popular uplink transmission strategy called over-the-air computation has emerged to support analog model uploading from massive edge devices \cite{zhu_broadband_2020}. Instead of orthogonal resource allocation among the devices to avoid interference, over-the-air computation allows the devices to share radio resources in model uploading by leveraging the signal-superposition property of analog transmission to conduct model aggregation over the air. To name a few, the authors in \cite{amiri_machine_2020,amiri_federated_2020} proposed gradient sparsification and compression prior to over-the-air transmission, where the model aggregation is reconstructed at ES via compressed sensing methods.

Besides, the distributive data among edge devices are typically neither independent nor identically distributed. This causes a major challenge called the statistical heterogeneity problem, which reduces the learning performance \cite{li_federated_2020}. To address this problem, the authors in \cite{smith_federated_2018} proposed the federated multi-task learning (FMTL) framework to implement multiple machine learning tasks over the FL communication network, motivated by the multi-task learning frameworks \cite{zhang_survey_2021,kumar_learning_2012}. Empirical results demonstrate that FMTL can significantly improve the generalizability of the trained models in the presence of statistical heterogeneity, where the knowledge contained in one task can be leveraged in training the models of the other tasks. Different from the method in  \cite{smith_federated_2018} only supporting a linear model or a linear combination of pre-trained models, the authors in \cite{dinh_fedu_2021,dinh_personalized_2021,li_ditto_2021} proposed other FMTL approaches to train more general non-convex learning models such as deep neural networks. To reduce the communication overhead, over-the-air computation is still an appealing solution to the uplink communication design of FMTL. Yet, the inter-task interference arises due to the concurrent transmissions of the model updates of multiple tasks. The inter-task interference, if not appropriately handled, will incur significant model aggregation errors that seriously degrade the learning performance.

To address the above challenge, in this paper, we propose the over-the-air FMTL (OA-FMTL) framework to achieve communication-efficient FMTL in the presence of inter-task interference, where multiple learning tasks are deployed on edge devices and a single ES. The ES is dedicated to parameter learning of multiple tasks, and the model aggregations of all the tasks are concurrently conducted at the ES via over-the-air computation. More specifically, at each edge device, the local model updates of all tasks are first sparsified and compressed individually by following the approach in \cite{sattler_sparse_2019,lin_deep_2020}, and then superimposed and sent over the uplink channel. At the ES, the model aggregations of all the tasks are reconstructed from the channel observation data, based on a modified version of the turbo compressed sensing (Turbo-CS) algorithm \cite{ma_turbo_2014} (named as M-Turbo-CS). It is clear that the above non-orthogonal transmission of multi-task model updates reduces the required uplink communication resource, but at the cost of severe inter-task interference. We show that our proposed OA-FMTL framework with M-Turbo-CS is able to efficiently suppress the inter-task interference, thereby significantly reducing the communication overhead, as compared to its counterpart with orthogonal multi-task transmission.

We further analyze the performance of the proposed OA-FMTL framework together with the M-Turbo-CS algorithm. Specially, we establish the state evolution to characterize the behaviour of the M-Turbo-CS algorithm. We then develop an upper bound on the learning loss of the OA-FMTL framework by taking model aggregation errors introduced by model sparsification, model misalignment, as well as imperfect reconstruction by M-Turbo-CS, into account. Based on that, we formulate a communication-learning optimization problem to improve the system performance by adjusting the power allocation among multiple tasks on the edge devices. We show that, under mild conditions, the problem can be solved by a feasibility test together with bisection search.

Numerical simulations show that our proposed OA-FMTL effectively suppresses the inter-task interference, so as to achieve a learning performance comparable to the single-task counterpart scheme \cite{amiri_machine_2020,amiri_federated_2020}. In other words, the communication overhead of OA-FMTL is only approximately one $N$-th of the conventional orthogonal scheme with independent training of each task, where $N$ is the total number of tasks. Meanwhile, simulations also show that our proposed optimization algorithm further reduces the communication overhead by appropriately adjusting the power allocation among the tasks.

The remainder of this paper is organized as follows. In Section \ref{sec:system model}, we describe the FMTL model and the over-the-air model aggregation framework. In Section \ref{sec:system}, we give the specific implementation of the OA-FMTL framework, including operations at devices, operations at the ES, and the design of the M-Turbo-CS algorithm. In Section \ref{sec:opti}, we analyze the performance of OA-FMTL and formulate an optimization problem that minimizes the training loss. In Section \ref{optimization}, we develop an effective algorithm to solve the optimization problem. In Section \ref{sec:num}, we validate our proposed OA-FMTL framework with experiments.

\emph{Notations}: Throughout, we use $\mathbb{R}$ and $\mathbb{C}$ to denote the real
and complex number sets, respectively. Regular letters, bold small letters, and bold capital letters are used to denote scalars, vectors, and matrices, respectively. We use $(\cdot)^T$ and $(\cdot)^H$ to denote the transpose and the conjugate transpose, respectively. We use $\mathcal{CN}(\mu, \sigma^2)$ to denote the circularly-symmetric complex normal distribution with mean $\mu$ and covariance $\sigma^2$, $||\cdot||$ to denote the $l_2$ norm, $||\cdot||_F$ to denote the Frobenius norm, $\bm{I}_N$ to denote the $N\times N$ identity matrix, $\operatorname{tr}(\cdot)$ to denote the trace of a square matrix, $\operatorname{inv}(\cdot)$ to denote the inverse of a square matrix, $[k]$ to denote the integer set $\{1,\dots,k\}$, $\mathbb{R}^+$ to denote the set of nonnegative real numbers, $\circ$ to denote the Hadamard product, $\bm{A}\succeq 0$ to denote that the matrix $\bm{A}$ is positive semidefinite, $\mathbb{E}[\cdot]$ to denote the expectation operator, and $\operatorname{var}[a|b]=\mathbb{E}[|a-\mathbb{E}[a|b]|^2|b]$. 

%\IEEEPARstart{T}{his} file is intended to serve as a ``sample article file''
%for IEEE journal papers produced under \LaTeX\ using
%IEEEtran.cls version 1.8b and later. The most common elements are covered in the simplified and updated instructions in ``New\_IEEEtran\_how-to.pdf''. For less common elements you can refer back to the original ``IEEEtran\_HOWTO.pdf''. It is assumed that the reader has a basic working knowledge of \LaTeX. Those who are new to \LaTeX \ are encouraged to read Tobias Oetiker's ``The Not So Short Introduction to \LaTeX ,'' available at: \url{http://tug.ctan.org/info/lshort/english/lshort.pdf} which provides an overview of working with \LaTeX.
\section{System Model}\label{sec:system model}
\subsection{Federated Multi-Task Learning}
We consider an FMTL system with $N$ learning tasks deployed on $M$ wireless local devices with the help of an ES, where task assignment to the devices is flexibly determined according to the computation power and storage capability of each device, as depicted in Fig. \ref{sec:system,ssec:multi,fig:system}. Each task $n$ on device $m$ is associated with its local dataset $D_{nm}$. Following \cite{zhang_convex_2012,smith_federated_2018,dinh_fedu_2021}, the FMTL problem is defined as 
\begin{subequations}\label{sec:system,ssec:multi,equ:optimum function}
	\begin{align}
		\min_{\bm{\Theta}, \bm{\Omega}}\quad&\mathcal{L}(\bm{\Theta})+\kappa_1\operatorname{tr}(\bm{\Theta}\bm{\Theta}^{T})+\kappa_2\operatorname{tr}(\bm{\Theta}\bm{\Omega}^{-1}\bm{\Theta}^{T})\label{fmtl obj}\\
		s.t.\quad&\bm{\Omega}\succeq 0,\label{fmtl cont1}\\ \quad&\operatorname{tr}(\bm{\Omega})=1,\label{fmtl cont2}
	\end{align}
\end{subequations}
where $\bm{\Theta}=[\bm{\theta}_1,\dots,\bm{\theta}_N]\in\mathbb{R}^{d\times N}$ models the parameters of the $N$ tasks with $\bm{\theta}_n$ being the parameter vector of task $n$ and $d$ being the common length of the model parameter vectors of all the tasks; $\bm{\Omega}\in\mathbb{R}^{N\times N}$ models the correlation between the $N$ tasks and is initialized to $\frac{1}{N}\bm{I}_N$; and $\kappa_1, \kappa_2$ are both the regularization parameters. In (\ref{fmtl obj}), the first term is the empirical learning loss defined as
\begin{equation}\label{total loss}
	\mathcal{L}(\bm{\Theta})=\sum_{n=1}^{N} \mathcal{L}_{n}(\bm{\theta}_{n}),
\end{equation}
where $\mathcal{L}_{n}(\bm{\theta}_{n})$ is the empirical learning loss of each task $n$, defined by
\begin{equation}\label{sec:system,ssec:multi,equ:loss function}	
	\mathcal{L}_n(\bm{\theta}_{n})=\sum_{m=1}^{M} \mathcal{L}_{nm}(\bm{\theta}_{n});
\end{equation}
the local empirical loss function of task $n$ on device $m$ is defined by 
\begin{equation}\label{sec:system,ssec:multi,equ:loss sample}	
	\mathcal{L}_{nm}(\bm{\theta}_{n})=\sum_{k=1}^{K_{nm}}l_n(\bm{\theta}_{n};\bm{u}_{nmk}),
\end{equation}  
where $l_n(\bm{\theta}_{n}; \bm{u}_{nmk})$ is the sample-wise loss function specified by task $n$, $\mathbf{u}_{nmk}$ denotes the $k$-th local data sample of dataset $D_{nm}$, and $K_{nm}$ is the cardinality of $D_{nm}$ with $K_{nm}=0$ meaning that $D_{nm}$ is empty. The second term of (\ref{fmtl obj}) penalizes the complexity of $\bm{\Theta}$, and the third term measures the relationships between all tasks based on $\bm{\Theta}$ and $\bm{\Omega}$ \cite{zhang_convex_2012}. The constraint (\ref{fmtl cont1}) is due to the fact that $\bm{\Omega}$ is defined as the task covariance matrix \cite{zhang_convex_2012}, and the constraint (\ref{fmtl cont2}) is to restrict the complexity of $\bm{\Omega}$.

\begin{figure}[h] 	
	\centering 	
	\includegraphics[width=0.85\linewidth]{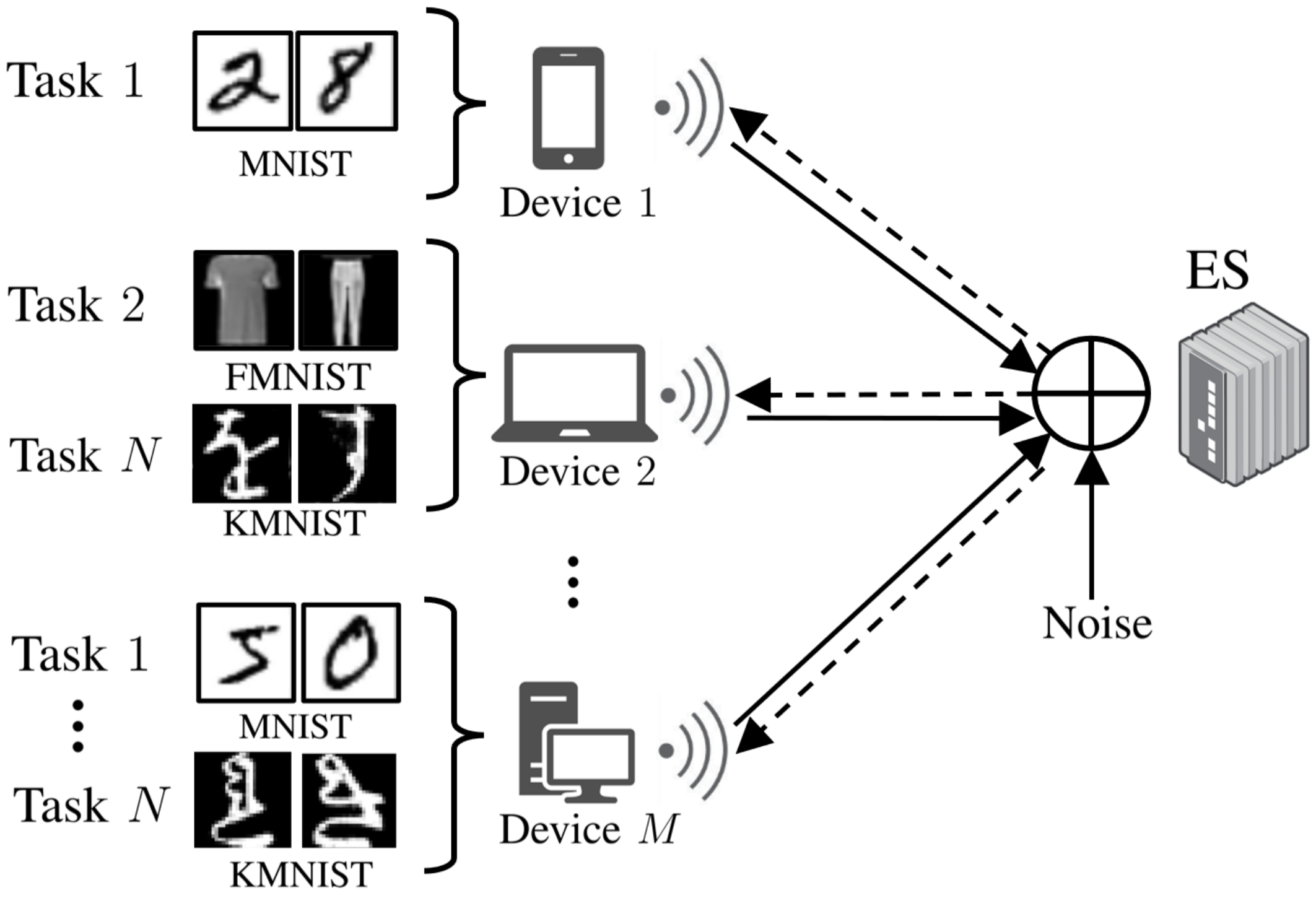} 	
	\caption{An illustration of the FMTL framework.} 	
	\label{sec:system,ssec:multi,fig:system}
\end{figure}

The minimization of (\ref{sec:system,ssec:multi,equ:optimum function}) is typically solved by the alternating descent method (ADM) \cite{zhang_convex_2012} in the FMTL framework over a wireless communication network. Specifically, at each communication round $t$, each device $m$ firstly uploads local gradients $\{\bm{g}_{nm}^{(t)}\}_{n=1}^N$ to the ES over a wireless uplink with $\bm{g}_{nm}^{(t)}=\nabla_n\mathcal{L}_{nm}(\bm{\theta}_n^{(t)})\in\mathbb{R}^d$ denoting the local gradient of task $n$ at device $m$, where $\nabla_n=\frac{\partial}{\partial \bm{\theta}_n}\in\mathbb{R}^d$ is the gradient operator with respect to $\bm{\theta}_n$. Then the ES optimizes the FMTL problem (\ref{sec:system,ssec:multi,equ:optimum function}) with respect to $\bm{\Theta}$ through a gradient-based update for a fixed $\bm{\Omega}$, i.e., the global model parameter $\bm{\Theta}$ is updated via
	\begin{align}\label{small update}
		\nonumber\bm{\Theta}^{(t+1)}= & \bm{\Theta}^{(t)}-\eta\left(\left[\sum_{m=1}^M\bm{g}_{1m}^{(t)},\dots,\sum_{m=1}^M\bm{g}_{Nm}^{(t)}\right]\right.\\
		&\left.+\ \kappa_1\bm{\Theta}^{(t)}+\kappa_2\bm{\Theta}^{(t)}{\bm{\Omega}^{(t)}}^{-1}\right),
	\end{align}	
where $\eta$ is a predetermined learning rate. With the definition of gradients $\bm{g}_{nm}^{(t)}$, (\ref{small update}) is equivalently rewritten as
\begin{equation}\label{update}
	\bm{\Theta}^{(t+1)}=\bm{\Theta}^{(t)}-\eta\left(\nabla \mathcal{L}(\bm{\Theta}^{(t)})+\kappa_1\bm{\Theta}^{(t)}+\kappa_2\bm{\Theta}^{(t)}{{\bm{\Omega}}^{(t)}}^{-1}\right),
\end{equation}
where $\nabla=\frac{\partial}{\partial \bm{\Theta}}\in\mathbb{R}^{d\times N}$ is the gradient operator with respect to $\bm{\Theta}$.
Since the FMTL problem (\ref{sec:system,ssec:multi,equ:optimum function}) is convex with respect to $\bm{\Omega}$ for a fixed $\bm{\Theta}$ \cite{zhang_convex_2012}, the ES updates the matrix $\bm{\Omega}$ via
\begin{equation}\label{update_relation}
	\bm{\Omega}^{(t+1)}=\frac{\left({\bm{\Theta}^{(t+1)}}^T\bm{\Theta}^{(t+1)}\right)^\frac{1}{2}}{\operatorname{tr}\left(\left({\bm{\Theta}^{(t+1)}}^T\bm{\Theta}^{(t+1)}\right)^{\frac{1}{2}}\right)}.
\end{equation}
After that, $\bm{\Theta}^{(t+1)}$ is broadcast to all the devices by the ES over a wireless downlink to synchronize the learning models among the devices. The iterative process in (\ref{update})-(\ref{update_relation}) continues until the learning tasks converge.

%Combining (\ref{sec:system,ssec:multi,equ:optimum function}) and (\ref{update}), the parameter segment $\bm{\theta}_{n}^{(t)}$ of each task $n$ is expected to be updated via
%\begin{subequations}\label{model_update} 
%	\begin{align}		\bm{\theta}_{n}^{(t+1)}&=\bm{\theta}_{n}^{(t)}-\eta\nabla_n \mathcal{L}_n(\bm{\theta}_{n})\\ 		\nonumber&=\bm{\theta}_{n}^{(t)} - \eta\frac{\sum_{m=1}^{M} K_{nm}\bm{g}_{nm}^{(t)}}{\sum_{m=1}^MK_{nm}}, \forall n \in [N],\tag{5b} 	
%	\end{align}
%\end{subequations}
%where the local gradient $\bm{g}_{nm}^{(t)}=\nabla_n \mathcal{L}_{nm}(\bm{\theta}_{n})\in\mathbb{R}^{d_n}$. 
%The following online groups are helpful to beginning and experienced \LaTeX\ users. A search through their archives can provide many answers to common questions.
%\begin{list}{}{}
%\item{\url{http://www.latex-community.org/}} 
%\item{\url{https://tex.stackexchange.com/} }
%\end{list}

\subsection{Over-the-Air Channel Model}
%See \cite{ref1,ref2,ref3,ref4,ref5} for resources on formatting math into text and additional help in working with \LaTeX .
We now describe the wireless channels used to support the above FMTL process. It is well known that the communication bottleneck of FL resides in the uplink gradient aggregation. Thus, following the convention in \cite{amiri_machine_2020,amiri_federated_2020,liu_reconfigurable_2021}, we assume error-free transmission in the downlink; and focus on the system design for the uplink. We model the wireless uplink as a block-fading multiple access (MAC) channel, where the channel state information (CSI) remains unchanged within $s$ channel uses. Since the update in (\ref{small update}) depends only on the sum of the local gradients, we employ the over-the-air computation for efficient model aggregation. Specifically, at the $t$-th communication round, the devices, each equipped with a signal antenna, send their signal to the ES over the block-fading MAC channel with $s$ channel uses, characterized by
\begin{equation}\label{sec:system,ssec:trans,equ:channel}
	\bm{r}^{(t)}=\sum_{m=1}^{M}h_m^{(t)}\bm{x}_m^{(t)} + \bm{w},	
\end{equation}
where $\bm{x}_m^{(t)}\in \mathbb{C}^s$ is the channel input vector from device $m$, $h_m^{(t)}\in\mathbb{C}$ is the channel gain from device $m$ to the ES, $\bm{r}^{(t)}\in \mathbb{C}^s$ is the channel output received by the ES, and $\bm{w} \in \mathbb{C}^s$ is an independent additive white Gaussian noise (AWGN) with each element independent and identically distributed as $\mathcal{CN}(0,\sigma_w^2)$. The power consumption of device $m$ at each round $t$ is constrained by 	\begin{equation}\label{sec:system,ssec:trans,equ:power constraint}
	\mathbb{E}[||x_{m,k}^{(t)}||^2] \leq P_m,
\end{equation}
where $x_{m,k}^{(t)}$ is the $k$-th entry of $\bm{x}_{m}^{(t)}$, and $P_m$ is the power budget of device $m$.

The remaining issue is to map the gradient vectors $\{\bm{g}_{nm}^{(t)}\}_{n=1}^N$ to the complex vector $\bm{x}_m^{(t)}$ at device $m$ in each communication round $t$, and to recover $\nabla \mathcal{L}(\bm{\Theta}^{(t)})$ in (\ref{update}) from $\bm{r}^{(t)}$ at the ES. As inspired by \cite{amiri_federated_2020,amiri_machine_2020,liu_reconfigurable_2021,ma_turbo_2014,seide_1-bit_2014}, we employ compressed sensing and error accumulation techniques to improve the communication efficiency. The details of the uplink transceiver design in the devices and the ES are presented in what follows.

\section{Proposed Transmission Scheme}\label{sec:system}
\subsection{Operations at Devices}\label{sec:system,ssec:trans}
To support the uplink transmission of the OA-FMTL framework, we process the local gradients of each task $n$ on device $m$, with the gradient sparsification scheme proposed in \cite{amiri_federated_2020} and a distinct gradient compression scheme proposed in this paper. Specifically, at each round $t$, device $m$ firstly adds $\bm{g}_{nm}^{(t)}$ defined above (\ref{small update}) with the error accumulation term $\bm{\Delta}_{nm}^{(t)}\in \mathbb{R}^{d}$ as
\begin{equation}\label{sec:system,ssec:trans,equ:error accumulated}
	\bm{g}_{nm}^{\operatorname{ac}(t)} = \bm{g}_{nm}^{(t)} + \bm{\Delta}_{nm}^{(t)},\forall m\in[M],\forall n\in[N],
\end{equation}
where $\bm{\Delta}_{nm}^{(t)}$ is accumulated in the previous rounds with $\bm{\Delta}_{nm}^{(1)}$ initialized to $\bm{0}$. Then device $m$ sets all the elements of $\bm{g}_{nm}^{\operatorname{ac}(t)}\in\mathbb{R}^{d}$ to zero except for the $k_n$ elements with the greatest magnitudes, defined by a mapping as
\begin{equation}\label{sec:system,ssec:trans,equ:sparse} 	
	\bm{g}_{nm}^{\operatorname{sp}(t)} = \operatorname{sp}(\bm{g}_{nm}^{\operatorname{ac}(t)}, k_n)\in\mathbb{R}^{d},\forall m\in[M],\forall n\in[N].
\end{equation}
$\bm{\Delta}_{nm}^{(t)}$ is updated by
\begin{equation}\label{sec:system,ssec:trans,equ:get error}
	\bm{\Delta}_{nm}^{(t+1)} = \bm{g}_{nm}^{\operatorname{ac}(t)} - \bm{g}_{nm}^{\operatorname{sp}(t)},\forall m\in[M],\forall n\in[N].
\end{equation}
Then $\bm{g}_{nm}^{\operatorname{sp}(t)}$ is normalized to $\bm{g}_{nm}^{\operatorname{no}(t)}\in\mathbb{R}^d$ by
\begin{equation}\label{normalization}
	\bm{g}_{nm}^{\operatorname{no}(t)}=\frac{\bm{s}_n\circ\bm{g}_{nm}^{\operatorname{sp}^{(t)}}}{v_{nm}^{(t)}},
\end{equation}
where ${v_{nm}^{(t)}}=\frac{1}{\sqrt{d}}||\bm{g}_{nm}^{\operatorname{sp}(t)}||$ is a normalization factor, and $\bm{s}_n\in\mathbb{R}^d$ is a random sign vector with the entries independently and uniformly drawn from $\{+1,-1\}$. Note that $\bm{s}_n$ ensures that the entries of $\bm{g}_{nm}^{\operatorname{no}(t)}$ have zero-mean, and $\bm{g}_{nm}^{\operatorname{no}(t)}$ of task $n$ is independent of the gradients of the other tasks. 

Next $\bm{g}_{nm}^{\operatorname{no}(t)}$ is compressed into a low-dimensional vector $\bm{g}_{nm}^{\operatorname{cp}(t)} \in \mathbb{R}^{2s}$ by a compression matrix $\bm{A}_n\in \mathbb{R}^{2s\times d}$ as
\begin{equation}\label{sec:system,ssec:trans,equ:compress}
	\bm{g}_{nm}^{\operatorname{cp}(t)} = \bm{A}_n\bm{g}_{nm}^{\operatorname{no}(t)},\forall m\in[M],\forall n\in[N],
\end{equation}
where each task $n$ is assigned with an individual compression matrix $\bm{A}_n$. Unlike the classic over-the-air FL framework \cite{amiri_federated_2020,amiri_machine_2020}, our proposed OA-FMTL framework needs to enable the concurrent reconstruction of the gradient aggregations of the $N$ tasks at the ES, where the randomness of the  compression matrices is introduced to help eliminate the inter-task interference. In specific, we employ a partial discrete cosine transform (DCT) matrix $\bm{A}_n=\bm{S}_n\bm{F}$ for each task $n$, where the selection matrix $\bm{S}_n\in \mathbb{R}^{2s\times d}$ consists of $2s$ randomly selected rows of the $d\times d$ identity matrix $\bm{I}_{d}$ and the $(m,n)$-th entry of the unitary DCT matrix $\bm{F}\in \mathbb{R}^{{d}\times {d}}$ is given by $\sqrt{\frac{2}{d}}\operatorname{cos}\left(\frac{(m-1)(2n-1)\pi}{2d}\right)$ when $m\neq1$, or $\sqrt{\frac{1}{d}}$ when $m=1$. Compared to other choices of the compression matrix such as the i.i.d. Gaussian matrix, the partial DCT matrix has advantages both in performance and complexity \cite{ma_performance_2015}.

We are now ready to describe the design of $\bm{x}_m^{(t)}$. As a distinct feature of the OA-FMTL framework, we propose to superimpose the local gradients of different tasks to support the multiplexing of the $N$ learning tasks. In specific, each device $m$ constructs
\begin{equation}\label{sec:system,ssec:trans,equ:mapping}
	\check{\bm{x}}_{m}^{(t)}=\sum_{n=1}^N \sqrt{\gamma_n^{(t)}}\bm{g}_{nm}^{\operatorname{cp}(t)}\in\mathbb{R}^{2s},
\end{equation}
where $\gamma_n^{(t)}\in\mathbb{R}^+$ is the inter-task power allocation coefficient of task $n$. Then $\check{\bm{x}}_{m}^{(t)}$ is converted into a complex vector $\tilde{\bm{x}}_{m}^{(t)}\in\mathbb{C}^{s}$, defined as
\begin{subequations}
	\begin{align}
		\operatorname{Re}\{\tilde{\bm{x}}_{m}^{(t)}\}&= \left[\check{x}_{m,1}^{(t)},\dots,\check{x}_{m,s}^{(t)}\right]^T, \\
		\operatorname{Im}\{\tilde{\bm{x}}_{m}^{(t)}\}&= \left[\check{x}_{m,s+1}^{(t)},\dots,\check{x}_{m,2s}^{(t)}\right]^T,
	\end{align}	
\end{subequations}	
where $\check{x}_{m,k}^{(t)}$ is the $k$-th entry of $\check{\bm{x}}_{m}^{(t)}$. After that, every device $m$ concurrently sends $\bm{x}_m^{(t)}=\alpha_m^{(t)}\tilde{\bm{x}}_{m}^{(t)}$ to the ES via analog transmission, where the power allocation cofficient $\alpha_m^{(t)}\in\mathbb{C}$ is given by
\begin{equation}\label{sec:system,ssec:trans,equ:power_allocation}
	\alpha_{m}^{(t)}=\frac{|h_m^{(t)}|}{h_m^{(t)}}\sqrt{\frac{P_m}{2}}.
\end{equation}
We see that the phase of $\alpha_m^{(t)}$ is opposite to the phase of $h_m^{(t)}$, and that with $\alpha_m^{(t)}$ in (\ref{sec:system,ssec:trans,equ:power_allocation}), the transmission power of every device meets the power budget. This avoids the hard alignment of local gradients and hence improves the system performance; see, e.g., \cite{zhong_over--air_2021} for more detailed discussions.

\subsection{Operations at ES}
We now describe the receiver design of the ES. We assume that the ES knows the CSI $\{h_m^{(t)}\}_{m=1}^M$ at each round $t$. In practice, the CSI can be acquired by the ES via channel training. With (\ref{sec:system,ssec:trans,equ:channel}), (\ref{normalization})-(\ref{sec:system,ssec:trans,equ:power_allocation}) and appropriate scaling, we construct the system model as
\begin{align}\label{sec:system,ssec:trans,equ:channel final}
	\bm{y}^{(t)} = 
	\begin{bmatrix}
		\bm{A}_1,\dots,\bm{A}_N
	\end{bmatrix}
	\begin{bmatrix}
		\bm{g}_1^{(t)^T},\dots,\bm{g}_N^{(t)^T}
	\end{bmatrix}^T + \bm{n},
\end{align}
where 
$\bm{y}^{(t)} =\frac{[\operatorname{Re}\{\bm{r}^{(t)}\}^T,\operatorname{Im}\{\bm{r}^{(t)}\}^T]^T}{\sqrt{\sum_{m=1}^M|h_m^{(t)}\alpha_m^{(t)}|^2}}$, $\bm{n} \triangleq \frac{[\operatorname{Re}\{\bm{w}\}^T,\operatorname{Im}\{\bm{w}\}^T]^T}{\sqrt{\sum_{m=1}^M|h_m^{(t)}\alpha_m^{(t)}|^2}}$ follows $\mathcal{N}(0,\sigma^2)$ with $\sigma^2 \triangleq \frac{\sigma_w^2}{2\sum_{m=1}^M|h_m^{(t)}\alpha_m^{(t)}|^2}$, and $\bm{g}_n^{(t)} \triangleq\frac{\sum_{m=1}^{M}\sqrt{\gamma_n^{(t)}}h_m^{(t)}\alpha_{m}^{(t)}\bm{g}_{nm}^{\operatorname{no}(t)}}{\sqrt{\sum_{m=1}^M|h_m^{(t)}\alpha_m^{(t)}|^2}}$. Given $\bm{y}^{(t)}$, for $\forall n\in [N]$, the ES reconstructs $\bm{g}_n^{(t)}$ as $\hat{\bm{g}}_n^{(t)}$ via the modified turbo compressed sensing algorithm detailed later in Section \ref{M-turbo-CS}. Then the ES computes the model updating vector $\check{\bm{g}}_n^{(t)}\in\mathbb{R}^d$ as
\begin{equation}\label{es con model updating}
	\check{\bm{g}}_n^{(t)}=\zeta_n^{(t)}\sqrt{\sum_{m=1}^M|h_m^{(t)}\alpha_m^{(t)}|^2}\bm{s}_n\circ\hat{\bm{g}}_n^{(t)}\in\mathbb{R}^d,
\end{equation}
where $\zeta_n^{(t)}\in\mathbb{R}$ is a normalization scaling factor of task $n$ to reduce the channel misalignment error caused by the incomplete channel reversion in (\ref{sec:system,ssec:trans,equ:power_allocation}), and the sign vector $\bm{s}_n$ defined below (\ref{normalization}) is known by the ES, e.g., via sharing the random seeds. Then, based on (\ref{update}), the model update of OA-FMTL in the presence of the communication error is given by
\begin{equation}\label{sec:system,ssec:multi,equ:paramenter update_new}
	\bm{\Theta}^{(t+1)}=\bm{\Theta}^{(t)}-\eta\left([\check{\bm{g}}_N^{(t)},\dots,\check{\bm{g}}_N^{(t)}]+\kappa_1\bm{\Theta}^{(t)}+\kappa_2\bm{\Theta}^{(t)}{\bm{\Omega}}^{(t)}\right),
\end{equation}
where $\eta$ is defined below (\ref{update}).

\subsection{Modified Turbo Compressed Sensing (M-Turbo-CS)}\label{M-turbo-CS}
The recovery of $\{\bm{g}_n^{(t)}\}_{n=1}^N$ from $\bm{y}^{(t)}$ in (\ref{sec:system,ssec:trans,equ:channel final}) is a compressed sensing problem, where the compression matrix $\bm{A}=[\bm{A}_1,\dots,\bm{A}_N]\in\mathbb{R}^{2s\times dN}$ composed of $N$ partial DCT matrices is partial orthogonal, i.e., $\bm{A}\bm{A}^T=\bm{I}_{2s}$. It is known that the Turbo-CS algorithm \cite{ma_turbo_2014} is the state-of-the-art compressed sensing problem to handle partial-orthogonal sensing matrices. Inspired by this, we basically follow the idea of Turbo-CS to solve the compressed sensing problem involved in (\ref{sec:system,ssec:trans,equ:channel final}). But the difference is that, as each $\bm{g}_n^{(t)}$ is the gradient for a different task $n$, $\{\bm{g}_n^{(t)}\}_{n=1}^N$ generally have different \emph{prior} distributions. With (\ref{normalization})-(\ref{sec:system,ssec:trans,equ:channel final}), we assume that the entries of $\bm{g}_n^{(t)}$ are independently drawn from a Bernoulli Gaussian distribution:
\begin{equation}\label{sec:system,ssec:recei,equ:Bernoulli Gaussian}
	g_{n,k}^{(t)} \sim\left\{\begin{array}{ll}
		0, & \text { probability }=1-\lambda_n^{(t)}, \\
		\mathcal{N}\left(0, \gamma_n^{(t)}\right), & \text { probability }=\lambda_n^{(t)},
	\end{array}\right.
\end{equation}
where $g_{n,k}^{(t)}$ is the $k$-th element of $\bm{g}_n^{(t)}$, $\lambda_n^{(t)}$ is the sparsity of $\bm{g}_n^{(t)}$, and $\gamma_n^{(t)}$ is the inter-task power allocation coefficient defined in (\ref{sec:system,ssec:trans,equ:mapping}). The sparsity $\lambda_n^{(t)}$ in the \emph{prior} distribution is estimated by the expectation-maximization algorithm \cite{vila_expectation-maximization_2013}.  With the above \emph{prior} model, we modify the Turbo-CS algorithm accordingly to accommodate concurrent model aggregations of the $N$ tasks as follows.

% Turbo-CS算法的好处和特点 	
\begin{figure}[h] 	\centering 	\includegraphics[width=1\linewidth]{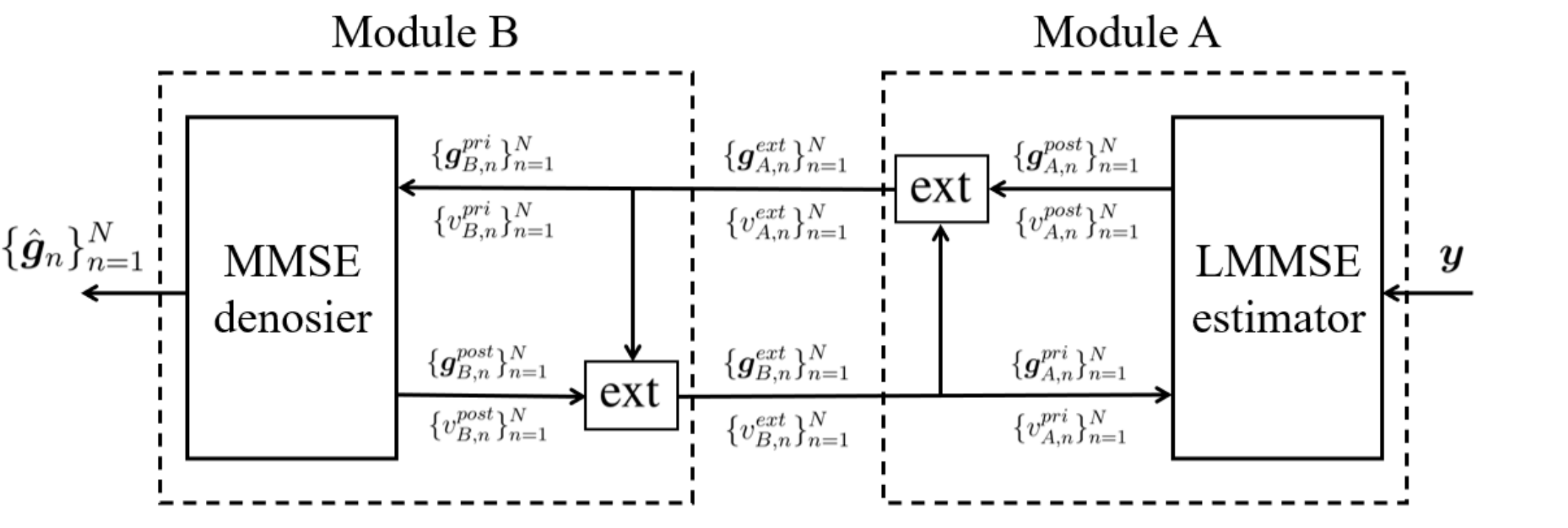} 	\caption{An illustration of the M-Turbo-CS algorithm.} 	\label{sec:system,ssec:recei,fig:Turbo-CS} \end{figure} 
As shown in Fig. \ref{sec:system,ssec:recei,fig:Turbo-CS}, M-Turbo-CS iterates between two modules where module A is a linear minimum mean-squared error (LMMSE) \cite{steven_fundamentals_1993} estimator handling the linear constraint in (\ref{sec:system,ssec:trans,equ:channel final}), and module B is a minimum mean-squared error (MMSE) \cite{steven_fundamentals_1993} denoiser exploiting the gradient sparsity in (\ref{sec:system,ssec:recei,equ:Bernoulli Gaussian}). The two modules are executed iteratively until convergence. For brevity, we henceforth drop out the round index $t$ in circumstances without causing ambiguity. The detailed operations of the iterative algorithm are provided below. 

The iterative process begins with Module A. Given the \emph{prior} mean $\bm{g}^{{pri}}_{A,n}\in\mathbb{R}^{d}$ initialized to $\bm{0}$, the \emph{prior} variance $v^{{pri}}_{A,n}\in\mathbb{R}^{d}$ initialized to $1$ and the observed vector $\bm{y}$ in (\ref{sec:system,ssec:trans,equ:channel final}) as the input of the LMMSE estimator in Module A, the \emph{posterior} mean $\bm{g}^{post}_{A,n}\in\mathbb{R}^{d}$ and the \emph{posterior} variance $v^{post}_{A,n}\in\mathbb{R}$ of the LMMSE estimator is given by
\begin{subequations}\label{sec:system,ssec:recei,equ:multi task turbo A}		
	\begin{align}
		\bm{g}_{A,n}^{post}&=\bm{g}_{A,n}^{pri}+\frac{\gamma_nv_{A,n}^{pri}(\bm{y}-\sum_{n'=1}^N\bm{A}_{n'}\bm{g}^{pri}_{A,n'})}{\sum_{n'=1}^N\gamma_{n'}v_{A,n'}^{pri}+\sigma^2}\bm{A}_n^T,\forall n\in[N], \label{sec:system,ssec:recei,equ:multi task turbo xA}\\
		v_{A,n}^{post}& = v_{A,n}^{pri} - \frac{s}{d}\frac{\gamma_n{v_{A,n}^{pri}}^2}{\sum_{n'=1}^N\gamma_{n'} v_{A,n'}^{pri}+\sigma^2},\forall n\in[N].\label{sec:system,ssec:recei,equ:multi task turbo vA}	
	\end{align}
\end{subequations}		
Then Module A outputs the extrinsic message of the LMMSE estimate as
\begin{subequations}\label{sec:system,ssec:recei,equ:multi task turbo B_PRI}	
	\begin{align}
		\bm{g}^{ext}_{A,n} &= v^{ext}_{A,n}
		\begin{pmatrix}
			\frac{\bm{g}^{post}_{A,n}}{v^{post}_{A,n}}-\frac{\bm{g}^{pri}_{A,n}}{v^{pri}_{A,n}}
		\end{pmatrix},\forall n\in[N], \label{sec:system,ssec:recei,equ:multi task turbo xB_PRI}\\
		v^{ext}_{A,n} &=
		\begin{pmatrix}
			\frac{1}{v^{post}_{A,n}}-\frac{1}{v^{pri}_{A,n}}
		\end{pmatrix}^{-1},\forall n\in[N],\label{sec:system,ssec:recei,equ:multi task turbo vB_PRI}
	\end{align}
\end{subequations}
which is used to update the \emph{prior} mean $\bm{g}^{{pri}}_{B,n}\in\mathbb{R}^{d}$ and the \emph{prior} variance $v^{{pri}}_{B,n}\in\mathbb{R}^{d}$ of the MMSE denoiser in Module B as
\begin{subequations}\label{sec:system,ssec:recei,equ:multi task turbo Aext}
	\begin{align}
		\bm{g}^{pri}_{B,n}&= \bm{g}^{ext}_{A,n},\label{sec:system,ssec:recei,equ:multi task turbo Aext mean} \\
		v^{pri}_{B,n}&=v^{ext}_{A,n}.\label{sec:system,ssec:recei,equ:multi task turbo Aext vari}
	\end{align}
\end{subequations}

For Module B, we model the \emph{prior} mean $\bm{g}_{B,n}^{pri}$ as an observation of $\bm{g}_{n}$ corrupted by additive noise $\bm{n}_{n}$ \cite{ma_turbo_2014}:
\begin{equation}\label{turbo base}
	\bm{g}_{B,n}^{pri} = \bm{g}_{n} + \bm{n}_{n}, 	
\end{equation}
where $\bm{n}_{n}\sim\mathcal{N}(0,\gamma_nv_{B,n}^{pri})$ is independent of $\bm{g}_{n}$. Then the \emph{posterior} mean $\bm{g}^{post}_{B,n}\in\mathbb{R}^{d}$ and the variance $v^{post}_{B,n}\in\mathbb{R}$ of the MMSE denoiser are given by
\begin{subequations}\label{sec:system,ssec:recei,equ:multi task turbo B}
	\begin{align}
		\bm{g}_{B,n}^{post} &= \mathbb{E}[\bm{g}_{n} \mid \bm{g}_{B,n}^{pri}],\forall n\in[N], \label{sec:system,ssec:recei,equ:multi task turbo xB} 	\\
		v_{B,n}^{post} &= \frac{1}{d\gamma_n}\sum_{k=1}^d
		\operatorname{var}[g_{n,k} \mid g_{B,n,k}^{pri}],\forall n\in[N],\label{sec:system,ssec:recei,equ:multi task turbo vB}
	\end{align}
\end{subequations}
where $\operatorname{var}[a|b]=\mathbb{E}[|a-\mathbb{E}[a|b]|^2|b]$; and $g_{n,k}$ and $g_{B,n,k}^{pri}$ are the $k$-th elements of $\bm{g}_{n}$ and $\bm{g}_{B,n}^{pri}$, respectively. Then, Module B outputs the extrinsic message of the MMSE denoiser as
\begin{subequations}\label{sec:system,ssec:recei,equ:multi task turbo A_pri}
	\begin{align}	
		\bm{g}^{ext}_{B,n} &=v^{ext}_{B,n}
		\begin{pmatrix}
			\frac{\bm{g}^{post}_{B,n}}{v^{post}_{B,n}}-\frac{\bm{g}^{pri}_{B,n}}{v^{pri}_{B,n}}
		\end{pmatrix},\forall n\in[N],\label{sec:system,ssec:recei,equ:multi task turbo xA_pri} 		\\
		v^{ext}_{B,n} &=
		\begin{pmatrix}
			\frac{1}{v^{post}_{B,n}}-\frac{1}{v^{pri}_{B,n}}
		\end{pmatrix}^{-1},\forall n\in[N]., \label{sec:system,ssec:recei,equ:multi task turbo vA_pri}	 	 	
	\end{align}
\end{subequations}
which is used to update the \emph{prior} mean $\bm{x}^{pri}_{A,n}\in\mathbb{R}^{d}$ and the variance $v^{pri}_{A,n}\in\mathbb{R}$ of the LMMSE estimator in Module A as
\begin{subequations}\label{sec:system,ssec:recei,equ:multi task turbo B ext}
	\begin{align}
		\bm{g}^{{pri}}_{B,n}&= \bm{g}^{ext}_{ A,n},\label{sec:system,ssec:recei,equ:multi task turbo B ext mean} \\
		v^{{pri}}_{B,n}&=v^{ext}_{ A,n}\label{sec:system,ssec:recei,equ:multi task turbo B ext vari}.
	\end{align}
\end{subequations}
At the end of the iterative process, the final estimate of $\bm{g}_n$ is based on the a \emph{posterior} output of the MMSE denoiser, i.e., $\hat{\bm{g}}_n=\bm{g}_{B,n}^{post}$.

To sum up, at each round $t$, given the initialization values $\bm{g}_{A,n}^{pri}=\bm{0}$ and $v_{A,n}^{pri}=1$ defined in (\ref{sec:system,ssec:recei,equ:Bernoulli Gaussian}) for $\forall n\in[N]$, (\ref{sec:system,ssec:recei,equ:multi task turbo A})-(\ref{sec:system,ssec:recei,equ:multi task turbo Aext}) and (\ref{sec:system,ssec:recei,equ:multi task turbo B})-(\ref{sec:system,ssec:recei,equ:multi task turbo B ext}) are iterated until a certain termination criterion is met, and $\bm{g}_{B,n}^{post}$ is output as $\hat{\bm{g}}_n^{(t)}$ to compute $\check{\bm{g}}_n^{(t)}$ for the model update in (\ref{sec:system,ssec:multi,equ:paramenter update_new}), for $\forall n\in[N]$. Compared with the original Turbo-CS algorithm in \cite{ma_turbo_2014}, the main difference is that each subvector $\bm{g}_n$ in (\ref{sec:system,ssec:trans,equ:channel final}) has its individual \emph{prior} distribution as in (\ref{sec:system,ssec:recei,equ:Bernoulli Gaussian}). Later, we will show by performance analysis and numerical experiments in Section \ref{turbo performance} that the M-Turbo-CS algorithm is able to efficiently suppress the inter-task interference in model aggregation.

\subsection{Overall Scheme}
The overall scheme is summarized in Algorithm \ref{sec:system,ssec:alogr,alg:algorithm}, where Lines 3-7 are executed at the devices, and Lines 8-19 are executed at the ES. Particularly, Line 15 of Algorithm \ref{sec:system,ssec:alogr,alg:algorithm} involves the optimization of this scheme, which will be elaborated later in Section \ref{optimization}. In what follows, we present the performance analysis of the proposed scheme.

%The contributions of our proposed algorithm are summarized as follows.
%\begin{itemize}
%	\item We novelly introduce over-the-air computation into the communication design of FMTL, and propose an OA-FMTL framework, where multiple learning tasks deployed on edge devices share a non-orthogonal fading channel under the coordination of an ES.
%	\item We propose an effective transmission scheme based on model sparsification and turbo compressed sensing, so as to overcome the inter-task interference, thereby achieving significant reduction in the total number of channel uses with only slight learning performance degradation.
%  	\item We proposed a M-Turbo-CS algorithm to accommodate the model aggregations of proposed OA-FMTL. Specifically, M-Turbo-CS assumes that the prior information of model aggregation is drawn from the segmented Bernoulli Gauss rather than Bernoulli Gauss as Turbo-CS assumes. 
%	\item We propose an effective algorithm to optimize the inter-task power allocation factors $\{\gamma_n^{(t)}\}_{n=1}^N$ and $\{\zeta_{n}\}_{n=1}^N$, where the updating in line \ref{step:op} of Algorithm \ref{sec:system,ssec:alogr,alg:algorithm} is elaborated later in Section \ref{optimization}.
%\end{itemize}

\begin{algorithm}[h]
	\caption{OA-FMTL with M-Turbo-CS.}\label{sec:system,ssec:alogr,alg:algorithm}  		
	\begin{algorithmic}[1] %这个1 表示每一行都显示数字 			
		\STATE \textbf{Initialize} $\bm{\Delta}_{nm}^{(1)}=\bm{0},\bm{\Omega}^{(1)}=\bm{I}_N, \gamma_n=\frac{1}{N}, \zeta_n=0, \forall n\in[N],m\in[M]$
		\FOR{$t = 1,2,\dots$}
		\STATE \textbf{Each device $m$ does in parallel:}	
		
		\STATE Compute $\{\bm{g}_{nm}^{(t)}\}_{n=1}^N$ with $\{D_{nm}\}_{n=1}^N$ and $\bm{\theta}^{(t)}$
		\STATE Compute $\{\bm{x}_m^{(t)}\}_{n=1}^N$ via (\ref{sec:system,ssec:trans,equ:error accumulated})-(\ref{sec:system,ssec:trans,equ:sparse}) and (\ref{normalization})-(\ref{sec:system,ssec:trans,equ:power_allocation})
		\STATE Compute $\{\bm{\Delta}_{nm}^{(t+1)}\}_{n=1}^N$ via (\ref{sec:system,ssec:trans,equ:get error})
		
		\STATE Send $\bm{s}_m^{(t)}$ to the ES
		over the channel in (\ref{sec:system,ssec:trans,equ:channel})
		\STATE \textbf{ES does:}
		\STATE Receive $\bm{r}^{(t)}$ via (\ref{sec:system,ssec:trans,equ:channel}) and compute $\bm{y}^{(t)}$ via (\ref{sec:system,ssec:trans,equ:channel final})
		\STATE \textbf{Initialize} $\bm{g}_{A,n}^{pri}=\bm{0}$, $v_{A,n}^{pri}=1, \forall n\in[N]$
		\REPEAT 
		\STATE Update $\{\bm{g}_{B,n}^{post}\}_{n=1}^N$, via (\ref{sec:system,ssec:recei,equ:multi task turbo A})-(\ref{sec:system,ssec:recei,equ:multi task turbo A_pri})
		\UNTIL{convergence}
		\STATE $\hat{\bm{g}}_n^{(t)}=\bm{g}_{B,n}^{post}, \forall n\in[N]$
		\STATE Update $\{\gamma_n^{(t)},\zeta_n^{(t)}\}_{n=1}^N$ via calling Algorithm \ref{optimiation algorithm} 
		\STATE $\check{\bm{g}}_n^{(t)}=\zeta_n^{(t)}\sqrt{\sum_{m=1}^M|h_m^{(t)}\alpha_m^{(t)}|^2}\bm{s}_n\circ\hat{\bm{g}}_n^{(t)}, \forall n\in[N]$
		\STATE Update $\bm{\Theta}^{(t)}$ via (\ref{sec:system,ssec:multi,equ:paramenter update_new})
		\STATE Update $\bm{\Omega}^{(t)}$ via (\ref{update_relation})
		\STATE Broadcast $\bm{\Theta}^{(t+1)}, \{\gamma_n^{(t+1)}\}_{n=1}^N$ to all the devices
		\ENDFOR
	\end{algorithmic} 	
\end{algorithm}
\section{Performance Analysis}\label{sec:opti}
\subsection{Performance Analysis of M-Turbo-CS}\label{turbo performance}
To start with, we first describe the performance analysis of the M-Turbo-CS algorithm via state evolution. The main idea of state evolution is to characterize the behavior of the Turbo-CS algorithm by the variance transfer functions of the two modules \cite{ma_performance_2015}. Specifically, the variance transfer functions of Module A are defined by
\begin{align}\label{stf_1}
	z_{n}=\phi_{n}\left(v_1,\dots,v_N;\gamma_1,\dots,\gamma_N\right), \forall n\in[N],
\end{align}
where $v_n$ denotes the \emph{prior} variance of LMMSE estimator, i.e., $v_n\triangleq v_{A,n}^{pri}$, $z_n$ denotes the reciprocal of the \emph{prior} variance of MMSE denoiser, i.e., $z_n\triangleq \frac{1}{v_{B,n}^{pri}}$; and the function $\phi_n$ maps $\{v_n\}_{n=1}^N$ into $z_n$ with the inter-task power allocation coefficients $\{\gamma_n\}_{n=1}^N$ known by the ES. With (\ref{sec:system,ssec:recei,equ:multi task turbo vA}), (\ref{sec:system,ssec:recei,equ:multi task turbo vB_PRI}) and (\ref{sec:system,ssec:recei,equ:multi task turbo Aext vari}), the expressions of functions $\{\phi_{n}\}_{n=1}^N$ are given by
\begin{align}\label{stf_2}
	\nonumber\phi_{n}&\left(v_1,\dots,v_N;\gamma_1,\dots,\gamma_N\right) \\ =&\left(\frac{d}{s\gamma_n}\begin{pmatrix}
		\sum_{{n'}=1}^N\gamma_{n'}v_{{n'}}+\sigma^2
	\end{pmatrix}-v_{n}\right)^{-1}, \forall n\in[N].
\end{align}	
In addition, the variance transfer functions of Module B are defined by
\begin{align}\label{stf_3}
	v_{n}&=\psi_{n}\left(z_{n}\right), \forall n\in[N],
\end{align}
where the function $\psi_n$ maps $z_n$ into $v_n$ for each $n$ with exploiting the \emph{prior} information in (\ref{sec:system,ssec:recei,equ:Bernoulli Gaussian}). With (\ref{sec:system,ssec:recei,equ:multi task turbo vB}), (\ref{sec:system,ssec:recei,equ:multi task turbo vA_pri}) and (\ref{sec:system,ssec:recei,equ:multi task turbo B ext vari}), we obtain the expression of $\{\psi_{n}\}_{n=1}^N$ as
\begin{align}\label{stf_4}
	\psi_{n}\left(z_{n}\right)&=\left(\frac{1}{\mathrm{mmse}_{n}(z_n)}-z_n\right)^{-1},\forall n\in[N],
\end{align}	
where $\mathrm{m m s e}_{n}(z_n)=\frac{1}{\gamma_n} \mathbb{E}\left[\operatorname{var}\left[\bm{g}_n \mid \bm{g}_n+\bm{n}_n\right]\right]$, $\bm{n}_{n}\sim\mathcal{N}(0,\gamma_n/z_n)$ is defined in (\ref{turbo base}) and $\bm{g}_n$ is defined in (\ref{sec:system,ssec:recei,equ:Bernoulli Gaussian}). The variance transfer functions (\ref{stf_1}) and (\ref{stf_3}) iteratively characterize the variances $\{v_n\}_{n=1}^N$, with each $v_n$ converging to $v_n^{\star}$ \cite{ma_performance_2015}. The fixed point $\{v_n^{\star}\}_{n=1}^N$ tracks the normalized output MSEs of the recovery on $\{\bm{g}_{n}\}_{n=1}^N$ in the M-Turbo-CS algorithm. We will show that the analysis agrees well with simulations in Section \ref{sec:num}.

\subsection{Communication Error Analysis of OA-FMTL}
We now analyze the communication error of the OA-FMTL framework. Based on (\ref{update}) and (\ref{sec:system,ssec:multi,equ:paramenter update_new}), we define the overall model updating error at the $t$-th round as 
\begin{equation}\label{overall_error}
	\bm{E}^{(t)}=\nabla \mathcal{L}(\bm{\Theta}^{(t)})-
	[\check{\bm{g}}_1^{(t)},\dots,\check{\bm{g}}_N^{(t)} ]\in\mathbb{R}^{d\times N}.
\end{equation}
We define the $n$-th columns of $\bm{E}^{(t)}$ as $\bm{e}_n^{(t)}$, which is also the model updating error from task $n$, characterized by
\begin{subequations}\label{sec:system,ssec:conve,equ:errorn_all}
	\begin{align}\label{sec:system,ssec:conve,equ:error_n}  		
		\bm{e}_{n}^{(t)} = & \ \nabla_n \mathcal{L}_n(\bm{\theta}_{n}^{(t)})-\check{\bm{g}}_n^{(t)}\\ 		
		\nonumber =& \ \underbrace{\sum_{m=1}^{M} \bm{g}_{nm}^{(t)}-\sum_{m=1}^{M} \bm{g}_{nm}^{\operatorname{sp}(t)}}_{\text{Sparsification error}}\ \\
		\nonumber&+\underbrace{\sum_{m=1}^{M} \bm{g}_{nm}^{\operatorname{sp}(t)}-\zeta_n^{(t)}\sqrt{\sum_{m=1}^M|h_m^{(t)}\alpha_m^{(t)}|^2}\bm{g}_n^{(t)}}_{\text{Misalignment error}}\\& +\underbrace{\zeta_n^{(t)}\sqrt{\sum_{m=1}^M|h_m^{(t)}\alpha_m^{(t)}|^2}(\bm{g}_n^{(t)}-\hat{\bm{g}}_n^{(t)})}_{\text{Estimation error from M-Turbo-CS}}\\=& \ \bm{e}_{n,1}^{(t)}+\bm{e}_{n,2}^{(t)}+\bm{e}_{n,3}^{(t)}\in\mathbb{R}^{d\times 1}, \label{error123}
	\end{align} 
\end{subequations}		
where $\bm{e}_{n,1}^{(t)}\in\mathbb{R}^d$ denotes the sparsification error caused by the step in (\ref{sec:system,ssec:trans,equ:sparse}), $\bm{e}_{n,2}^{(t)}\in\mathbb{R}^d$ denotes the channel misalignment error caused by the incomplete channel reversion in (\ref{sec:system,ssec:trans,equ:power_allocation}), and $\bm{e}_{n,3}^{(t)}\in\mathbb{R}^d$ denotes the estimation error caused by the imperfect recovery of M-Turbo-CS. We analyze the bound of $\mathbb{E}[||\bm{E}^{(t)}||^2]$ as
\begin{align}\label{sec:system,ssec:conve,equ:error}		
	\mathbb{E}[||\bm{E}^{(t)}||_F^2] =\sum_{n=1}^N\mathbb{E}[||\bm{e}_n^{(t)}||^2],
\end{align} 
where $\mathbb{E}[||\bm{e}_n^{(t)}||^2]$ is bounded as
\begin{align}\label{expectation error}
	\mathbb{E}[||\bm{e}_n^{(t)}||^2] \leq  3(||\bm{e}_{n,1}^{(t)}||^2+\mathbb{E}[||\bm{e}_{n,2}^{(t)}||^2]+\mathbb{E}[||\bm{e}_{n,3}^{(t)}||^2]),
\end{align}
by using the triangle inequality and the inequality of arithmetic means. 
The analysis of $||\bm{e}_{n,1}^{(t)}||^2$ defined in (\ref{error:e1}) and $\mathbb{E}[||\bm{e}_{n,2}^{(t)}||^2]$ defined in (\ref{error:e2}) basically follows the process in \cite{amiri_machine_2020} and \cite{liu_reconfigurable_2021}, respectively, and is given in Appendix \ref{proof:Theorem 1}. Besides, from the performance analysis of M-Turbo-CS in Section \ref{turbo performance}, M-Turbo-CS converges to the fixed point $\{{v_n^\star}^{(t)}\}_{n=1}^N$ at the $t$-th round, for $\forall n\in[N]$. Thus, we have
\begin{equation}\label{turbo error}
	\mathbb{E}[||\bm{e}_{n,3}^{(t)}||^2]=\sum_{m=1}^M|h_m^{(t)}\alpha_m^{(t)}|^2{\zeta_n^{(t)}}^2\gamma_n^{(t)}{v_n^{\star}}^{(t)}.
\end{equation}
\subsection{Convergence Analysis}
To proceed, we make assumptions by following the convention in stochastic optimization \cite{friedlander_hybrid_2012}.
\begin{assumption}\label{sec:conver,asu:assumption 1}
	$\mathcal{L}_n(\cdot)$ is strongly convex with some (positive) parameter $\Omega_n$. That is, $\mathcal{L}_n(\bm{y}) \geq \mathcal{L}_n(\bm{x})+(\bm{y}-\bm{x})^{T} \nabla_n \mathcal{L}_n(\bm{x})+\frac{\Omega_n}{2} \| \bm{y}-\bm{x} \|^{2}, \forall \bm{x}, \bm{y} \in \mathbb{R}^{d},\forall n\in[N]$.
\end{assumption}  	 	\begin{assumption}\label{sec:conver,asu:assumption 2}
	The gradient $\nabla_n \mathcal{L}_n(\cdot)$ is Lipschitz continuous with some (positive) parameter $L$. That is, $\|\nabla_n \mathcal{L}_n(\bm{x})-\nabla_n \mathcal{L}_n(\bm{y})\| \leq L_n\|\bm{x}-\bm{y}\|, \forall \bm{x}, \bm{y} \in \mathbb{R}^{d},\forall n\in[N]$.
\end{assumption}  	 	\begin{assumption}\label{sec:conver,asu:assumption 3} 			
	$\mathcal{L}_n(\cdot)$ is twice-continuously differentiable, for $\forall n\in[N]$.
\end{assumption}  	 
\begin{assumption}\label{sec:conver,asu:assumption 4}
	The gradient with respect to any training sample, denoted by $\nabla_n l_n(\bm{\theta}_{n}; \cdot)$, is upper bounded at $\bm{\theta}_{n}$ as 	$$\nonumber\left\|\nabla_n l_n(\bm{\theta}_{n},\bm{u}_{nmk})\right\|^{2} \leq \beta_{n,1}+ \beta_{n,2} \left\|\nabla_n \mathcal{L}_n\left(\bm{\theta}_n\right)\right\|^{2},\forall n\in[N]$$ for some constants $\beta_{n,1} \geq 0$ and $\beta_{n,2}>0$.
\end{assumption}

Assumptions \ref{sec:conver,asu:assumption 1}-\ref{sec:conver,asu:assumption 4} lead to an upper bound on the loss function $\mathcal{L}_n(\bm{\theta}_n^{(t+1)})$ with respect to the model updating (\ref{sec:system,ssec:multi,equ:paramenter update_new}), as given in the following lemma.
\begin{lemma}\label{lemma:update} 
	Let $\mathcal{L}_n(\cdot)$ satisfy Assumptions \ref{sec:conver,asu:assumption 1}-\ref{sec:conver,asu:assumption 4}. At the $t$-th training round, with $L_n =1/\eta$for $\forall n$, we have
	\begin{align} \label{equ:lemma}
		\nonumber\mathbb{E}[\mathcal{L}_n(\bm{\theta}_n^{(t+1)})] \leq&\   	\mathbb{E}[\mathcal{L}_n(\bm{\theta}_n^{(t)})]-\frac{1}{2L_n}	\mathbb{E}[\|\nabla_n \mathcal{L}_n(\bm{\theta}_n^{(t)})\|^{2}]\\
		&+\frac{1}{2L_n}	\mathbb{E}[ \|\bm{e}_n^{(t)}\|^{2}],
	\end{align}
	where the Lipschitz constant $L_n$ is defined in Assumption \ref{sec:conver,asu:assumption 2}.
\end{lemma}
\begin{proof}
	See \cite[Lemma 2.1]{friedlander_hybrid_2012}.
\end{proof}

We are now ready to derive an upper bound of the expected difference between the training loss and the optimal loss at round $t+1$, i.e., $\mathbb{E}[\mathcal{L}(\bm{\Theta}^{(t+1)})-\mathcal{L}(\bm{\Theta}^{(t)})]$. 
\begin{thm}\label{thm:1}
	With Assumptions \ref{sec:conver,asu:assumption 1}-\ref{sec:conver,asu:assumption 4},
	\begin{align}\label{equ:final_convergence}
		\nonumber\mathbb{E}[\mathcal{L}(\bm{\Theta}^{(t+1)})-\mathcal{L}(\bm{\Theta}^{(t)})]\leq&
		\mathbb{E}[\mathcal{L}(\bm{\Theta}^{(1)})-\mathcal{L}(\bm{\Theta}^{(\star)})](\max_n\Upsilon_n)^t\\
		&+\sum_{t'=1}^t(\max_n\Upsilon_n)^{t-t'}\sum_{n=1}^NC_n^{(t')},
	\end{align}
	where $(\cdot)^{t}$ denotes the $t$-th power, $\mathcal{L}(\cdot)$ is the total empirical loss function defined in (\ref{sec:system,ssec:multi,equ:optimum function}), $\bm{\Theta}^{(1)}$ is the initial system model parameter, and the functions $\Upsilon_n$, $C_n^{(t)}$ for each task $n$ are defined as
	\begin{subequations}\label{define:psi,upsilon_C}
		\begin{align}
			\Upsilon_n\triangleq1-\frac{\Omega_n}{L_n}+\frac{3\Omega_n\beta_{n,2}}{L_n}\left(\frac{2r_n}{1-r_n}\right)^2,\\
			C_n^{(t)}\triangleq\frac{3}{2L_n}\left(\beta_{n,1}\left(\frac{2r_n}{1-r_n}\right)^2+\Psi_n^{(t)}\right),
		\end{align}
	\end{subequations}
	where $r_n=\sqrt{(d_n-k_n)/d_n}<1$ with $k_n$ defined above (\ref{sec:system,ssec:trans,equ:sparse}), and the function $\Psi_n^{(t)}$ for each task $n$ is defined by
	\begin{align}\label{define:psi}
		\nonumber\Psi_n^{(t)}\triangleq&\sum_{m=1}^M\left(v_{nm}^{(t)}-\zeta_n^{(t)}\sqrt{\gamma_n^{(t)}}|h_m^{(t)}\alpha_{m}^{(t)}|\right)^2\\&+\sum_{m=1}^M|h_m^{(t)}\alpha_m^{(t)}|^2{\zeta_n^{(t)}}^2\gamma_n^{(t)}{v_n^{\star}}^{(t)},
	\end{align}
	where $h_m^{(t)}$ is the channel gain from device $m$ defined in (\ref{sec:system,ssec:trans,equ:channel}), $v_{nm}^{(t)}$ is defined in (\ref{normalization}), $\gamma_n^{(t)}$ is the inter-task power allocation coefficient of task $n$ defined in (\ref{sec:system,ssec:trans,equ:mapping}), $\alpha_m^{(t)}$ is the power allocation coefficient defined in (\ref{sec:system,ssec:trans,equ:power_allocation}), and $\zeta_n^{(t)}$ is a normalization scaling factor of task $n$ defined in (\ref{es con model updating}).
\end{thm}
\begin{proof}
	See Appendix \ref{proof:Theorem 1}.
\end{proof}

From Theorem \ref{thm:1}, we see that $\mathbb{E}[\mathcal{L}(\bm{\Theta}^{(t)})-\mathcal{L}(\bm{\Theta}^{(\star)})]$, the expected difference between the training loss and the optimal loss at the $t$-th round, is upper bounded by the right-hand side of the inequality in (\ref{equ:final_convergence}). Moreover, this upper bound converges with speed $\max_n\Upsilon_n$ when $\max_n\Upsilon_n<1$. Empirically, we find that the proposed OA-FMTL scheme always converges with appropriately chosen system parameters. 
The following corollary further characterizes the convergence behavior of $\mathbb{E}[\mathcal{L}(\bm{\Theta}^{(t+1)})-\mathcal{L}(\bm{\Theta}^{(\star)})]$.

\begin{coro}\label{coro1}
	Suppose that Assumptions \ref{sec:conver,asu:assumption 1}-\ref{sec:conver,asu:assumption 4} hold, and that $\max_n\Upsilon_n < 1$ for $\forall n$. Then, as $t\rightarrow\infty$, we have 
	\begin{equation}\label{final_convegence_gap}
		\lim_{t \rightarrow \infty}\mathbb{E}[\mathcal{L}(\bm{\Theta}^{(t+1)})-\mathcal{L}(\bm{\Theta}^{(\star)})]\leq \sum_{t'=1}^{t}\sum_{n=1}^NC_n^{(t')}
	\end{equation}
\end{coro}
\begin{proof}
	When $\max_n\Upsilon_n < 1$, we have $\lim_{t\rightarrow\infty}(\max_n\Upsilon_n)^t=0$. Plugging this result and the assumption into (\ref{equ:final_convergence}), we obtain (\ref{final_convegence_gap}).
\end{proof}

Corollary \ref{coro1} shows that OA-FMTL guarantees to converge with a sufficiently small $\max_n\Upsilon_n$. However, there generally exists a gap between the converged loss $\lim_{t \rightarrow \infty}\mathbb{E}[\mathcal{L}(\bm{\Theta}^{(t+1)})]$ and the optimal one $\mathbb{E}[\mathcal{L}(\bm{\Theta}^{(\star)})]$ because of the communication noise and the inter-task interference. Therefore, we next aim to minimize the gap, so as to optimize our proposed framework.
\section{Optimization}\label{optimization}
\subsection{Problem Formulation}\label{Problem Formulation}
From Theorem \ref{thm:1} and Corollary \ref{coro1}, we see that the term $\sum_{t^\prime=1}^{t}\sum_{n=1}^NC_n^{(t^\prime)}$ in (\ref{final_convegence_gap}) represents the impact of the misalignment error and the estimation error caused by M-Turbo-CS on the convergence rate and the asymptotic learning performance. Specifically, a smaller $\sum_{t'=1}^{t}\sum_{n=1}^NC_n^{(t)}$ leads to faster convergence and a smaller gap in $\lim_{t \rightarrow \infty}\mathbb{E}[\mathcal{L}(\bm{\Theta}^{(t+1)})-\mathcal{L}(\bm{\Theta}^{(\star)})]$. This motivates us to use $\sum_{t'=1}^{t}\sum_{n=1}^NC_n^{(t')}$ as the metric of the FL performance, i.e., to minimize $\sum_{n=1}^NC_n^{(t)}$ over $\{\{\gamma_n^{(t)}\}_{n=1}^N, \{\zeta_n^{(t)}\}_{n=1}^N\}$ at each round $t$. We drop out the round index $t$ in the following for brevity, and formulate the design problem as
\begin{subequations}\label{op1}
	\begin{align}
		\mathcal{P}_{1}: \min _{\{\zeta_{n},\gamma_{n}\}_{n=1}^{N}} & \sum_{n=1}^{N}\sum_{m=1}^{M}\left(v_{nm}-\zeta_n\sqrt{\gamma_n}|h_m\alpha_{m}|\right)^2\label{p1_objec}\\&\nonumber+\sum_{n=1}^{N}{\sum_{m=1}^M|h_m\alpha_m|^2{\zeta_n^2}\gamma_nv_n^{\star}(\gamma_1,\dots,\gamma_N)}, \\
		\text { s.t. } \quad & \gamma_{n} > 0, \quad \forall n \in[N],\label{p1_con1} \\
		& \sum_{n=1}^{N} \gamma_{n}\leq1,\label{p1_inter_task}
	\end{align}
\end{subequations}
where each $v_n^{\star}$, which is in general a function of $\{\gamma_n\}_{n=1}^N$, is the fixed point of the state evolution in (\ref{stf_1}) and (\ref{stf_3}).
In problem $\mathcal{P}_{1}$, the objective function (\ref{p1_objec}) is equivalent to $\sum_{n=1}^NC_n$ in (\ref{define:psi,upsilon_C}) by dropping the term irrelevant to $\{\{\gamma_n\}_{n=1}^N, \{\zeta_n\}_{n=1}^N\}$; the constraint (\ref{p1_inter_task}) is obtained by substituting (\ref{normalization})-(\ref{sec:system,ssec:trans,equ:power_allocation}) into the power constraint (\ref{sec:system,ssec:trans,equ:power constraint}).

It is difficult to solve Problem $\mathcal{P}_{1}$ directly since the fixed point $\{v_n^{\star}\}_{n=1}^N$ of the state evolution in (\ref{stf_1}) and (\ref{stf_3}) does not have a closed-form expression. To simplify the problem, we introduce a constraint $v_n^\star(\gamma_1,\dots,\gamma_N)\leq v_n^{\text{target}}$, $\forall n$, with $v_n^{\text{target}}$ being a slack variable satisfying $0\leq v_n^{\text{target}}\leq 1$, yielding a modified problem as
\begin{subequations}\label{opti2}
	\begin{align}
		\nonumber\mathcal{P}_{2}: \min _{\{\zeta_{n},\gamma_{n},v_{n}^{\text{target}}\}_{n=1}^{N}} & \sum_{n=1}^{N}\sum_{m=1}^{M}\left(v_{nm}-\zeta_n\sqrt{\gamma_n}|h_m\alpha_{m}|\right)^2\\&+\sum_{n=1}^{N}{\sum_{m=1}^M|h_m\alpha_m|^2{\zeta_n^2}\gamma_nv_n^{\text{target}}}, \\
		\text { s.t. } \quad & \gamma_{n} > 0, \quad \forall n \in[N], \\
		& \sum_{n=1}^{N} \gamma_{n}\leq1,\\
		& v_n^\star(\gamma_1,\dots,\gamma_N)\leq v_n^{\text{target}},  \forall n \in[N],\label{p2_vstar}\\
		&0\leq v_n^{\text{target}}\leq 1.
	\end{align}
\end{subequations}
By inspection, it is not difficult to see that $\mathcal{P}_{2}$ has the same solution as $\mathcal{P}_{1}$. However, the fixed point $\{v_n^{\star}\}_{n=1}^N$ in the constraint (\ref{p2_vstar}) still does not have a closed-form expression. We next give an approximate but explicit expression of the constraint (\ref{p2_vstar}). To this end, we first define a simplified state transfer function as 
\begin{equation}\label{st_function_1}
	z_{n}=\hat{\phi}_{n}\left(v_n;\gamma_1,\dots,\gamma_N\right), \forall n\in[N],
\end{equation}
where the function $\hat{\phi}_{n}$, only mapping $v_n$ into $z_n$, is a simplified version of $\phi_{n}$ defined in (\ref{stf_1}). Assume that $c_1v_{1}=\cdots=c_Nv_{n}$ in (\ref{stf_1}) with predetermined constants $\{c_n|c_n\in\mathbb{R}^+\}_{n=1}^N$. The expression of the function $\hat{\phi}_{n}$ is given by
\begin{align}\label{phi_n hat}
	\nonumber\hat{\phi}_{n}&\left(v_n;\gamma_1,\dots,\gamma_N\right)\\ =&\left(\left(\frac{dc_n}{ s\gamma_{n}}\sum_{{n'}=1}^N\frac{\gamma_{n'}}{c_{n'}}-1\right) v_{n}+\frac{d \sigma^{2}}{s\gamma_{n} }\right)^{-1},\forall n\in[N].
\end{align}
The recursion in (\ref{stf_4}) and (\ref{phi_n hat}) continues and converges to the fixed point $v_n^\star$ for $\forall n\in[N]$, satisfying
\begin{subequations}
	\begin{align}\label{transfer}
		\hat{\phi}_n(v_n^\star;\gamma_1,\dots,\gamma_N)&=\psi_n^{-1}(v_n^\star),\\ \hat{\phi}_n(v_n;\gamma_1,\dots,\gamma_N)&\geq\psi_n^{-1}(v_n),\forall v_n\in(v_n^\star,1],
	\end{align}
\end{subequations}
where $\psi_n^{-1}(\cdot)$ is the inverse of $\psi_n(\cdot)$ which exists since $\psi_n(\cdot)$ is continuous and monotonic \cite{ma_performance_2015}. Note that $v_n \leq 1$ since the signal power is normalized. The above recursive process is illustrated by Fig. \ref{fig:evolution}. With (\ref{transfer}), the inequality in (\ref{p2_vstar}) holds when $\hat{\phi}_{n}\left(v_n;\gamma_1,\dots,\gamma_N\right)\geq \psi_{n}^{-1}\left(v_{n}\right)$, where $v_n\in(v_n^{\text{target}},1]$ for $\forall n\in[N]$, yielding
\begin{subequations}\label{opti3}
	\begin{align}
		\nonumber\mathcal{P}_{3}: \min _{\{\zeta_{n},\gamma_{n},v_{n}^{\text{target}}\}_{n=1}^{N}} & \sum_{n=1}^{N}\sum_{m=1}^{M}\left(v_{nm}-\zeta_n\sqrt{\gamma_n}|h_m\alpha_{m}|\right)^2\\&+\sum_{n=1}^{N}{\sum_{m=1}^M|h_m\alpha_m|^2{\zeta_n^2}\gamma_nv_n^{\text{target}}}, \label{p3_obj}\\
		\text { s.t. } \quad & \gamma_{n} > 0, \quad \forall n \in[N], \\
		& \sum_{n=1}^{N} \gamma_{n}\leq1,\\
		& \nonumber\hat{\phi}_{n}\left(v_n;\gamma_1,\dots,\gamma_N\right)\geq \psi_{n}^{-1}\left(v_{n}\right),\\&\forall v_n\in(v_n^{\text{target}},1], \forall n \in[N],\label{p3_vstar}\\
		&0\leq v_n^{\text{target}}\leq 1.
	\end{align}
\end{subequations}
What remains is to solve $\mathcal{P}_3$, as detailed in the next subsection.
\begin{figure}
	\centering
	\includegraphics[width=0.85\linewidth]{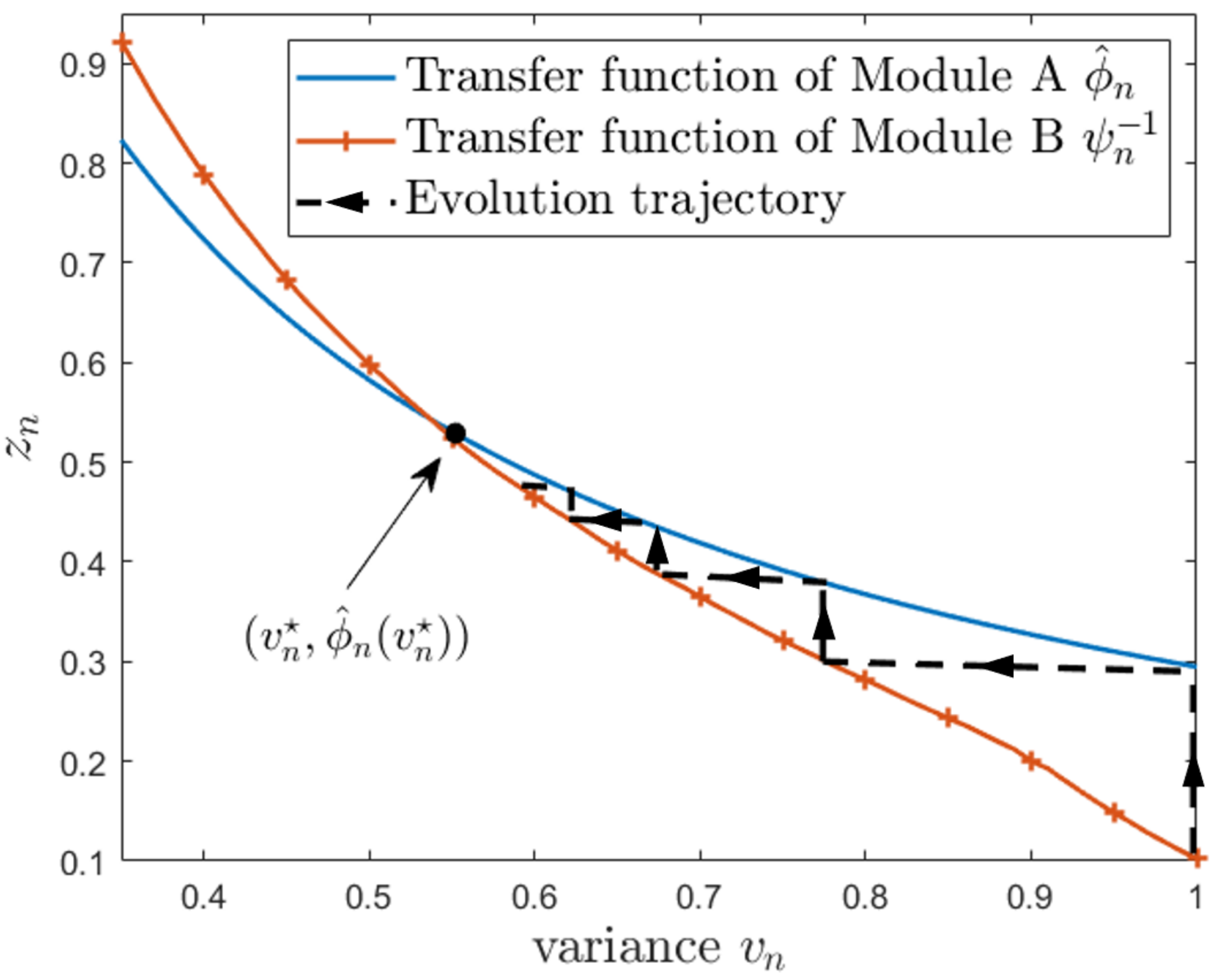}
	\caption{An illustration of the variance transfer.}
	\label{fig:evolution}
\end{figure}

\subsection{Algorithm Design}
In this subsection, we first give an equivalent form of Problem $\mathcal{P}_{3}$. We then show that the equivalent problem $\mathcal{P}_{4}$ is readily solved by using existing optimization tools. The details are provided in the following propositions.
\begin{prop}\label{prop1}
	Problem $\mathcal{P}_3$ is equivalent to
	\begin{subequations}\label{opti4}
		\begin{align}
			\nonumber\mathcal{P}_{4}: \min _{\{\gamma_{n},v_{n}^{\mathrm{target}}\}_{n=1}^{N}} & \sum_{n=1}^{N}\left(\sum_{m=1}^{M}		\left(v_{nm}-|h_m\alpha_{m}|\it{\Gamma}_{n}\right)^2\right.\\&\left.+\sum_{m=1}^{M}\left(\sqrt{v_n^{\mathrm{target}}}|h_m\alpha_{m}|\it{\Gamma}_{n}\right)^2\right),\label{opti4 obj}\\
			\text { s.t. } \quad & \gamma_{n} > 0, \quad \forall n \in[N], \\
			& \sum_{n=1}^{N} \gamma_{n}\leq1,\\
			& \nonumber\hat{\phi}_{n}\left(v_n;\gamma_1,\dots,\gamma_N\right)\geq \psi_{n}^{-1}\left(v_{n}\right),\\&\forall v_n\in(v_n^{\mathrm{target}},1], \forall n \in[N],\label{opti4_psi}\\
			&0\leq v_n^{\mathrm{target}}\leq 1,
		\end{align}
	\end{subequations}
	where ${\it{\Gamma}}_{n}\in\mathbb{R}^+$ is defined as
	\begin{equation}\label{Gamma}
		\it{\Gamma}_{n}=\sqrt{\gamma_n}\zeta_n=\frac{\sum_{{m'}=1}^Mv_{n{m'}}|h_{m'}\alpha_{m'}|}{(1+v_n^{\mathrm{target}})\sum_{m'=1}^{M}|h_{m'}\alpha_{m'}|^2}.
	\end{equation} 
	
\end{prop}
\begin{proof}
	For fixed $\{\gamma_{n}, v_{n}^{\text{target}}\}_{n=1}^{N}$, the problem $\mathcal{P}_{3}$ for finding $\{\zeta_n\}_{n=1}^N$ is a quadratic convex problem. Thus, setting the derivative of (\ref{p3_obj}) with respect to $\zeta_n$ to zero, we obtain
	\begin{equation}\label{zeta_solve1}
		\zeta_n=\frac{\sum_{m=1}^Mv_{nm}|h_m\alpha_m|}{\sqrt{\gamma_n}\sum_{m=1}^{M}|h_m\alpha_m|^2(1+v_n^{\mathrm{target}})},\forall n\in[N].
	\end{equation}
	By plugging (\ref{zeta_solve1}) into (\ref{p3_obj}), we obtain the objective of $\mathcal{P}_4$, which completes the proof.
\end{proof}

\begin{prop}\label{prop target}
	For fixed $\{\gamma_n\}_{n=1}^N$, the objective function of Problem $\mathcal{P}_{4}$ is monotonically increasing with respect to $v_n^{\mathrm{target}}$ for each $n$.
\end{prop}
\begin{proof}
	For fixed $\{\gamma_n\}_{n=1}^N$, the derivative of the first term $\sum_{m=1}^{M}		\left(v_{nm}-|h_m\alpha_{m}|\it{\Gamma}_{n}\right)^2$ in (\ref{opti4 obj}) with respect to each $v_n^{\text{target}}$ is given by
	\begin{subequations}
		\begin{align}
			&\frac{\mathrm{d} \sum_{m=1}^{M}\left(v_{nm}-|h_m\alpha_{m}|\it{\Gamma}_{n}\right)^2}{\mathrm{d}v_n^{\text{target}}}\\
			=&\frac{\mathrm{d} \sum_{m=1}^{M}\left(v_{nm}-|h_m\alpha_{m}|\it{\Gamma}_{n}\right)^2}{\mathrm{d}\it{\Gamma}_{n}}\frac{\mathrm{d}\it{\Gamma}_{n}}{\mathrm{d}v_n^{\text{target}}}\\
			=&\frac{2v_n^{\text{target}}}{(1+v_n^{\text{target}})^3}\frac{\left(\sum_{m=1}^Mv_{nm}|h_m\alpha_m|\right)^2}{\sum_{m=1}^{M}|h_{m}\alpha_{m}|^2},
		\end{align}
	\end{subequations}
	which is always positive. Together with the fact that $\sum_{m=1}^{M}\left(\sqrt{v_n^{\text{target}}}|h_m\alpha_{m}|\it{\Gamma}_{n}\right)^2$ in (\ref{opti4 obj}) increases monotonically in $v_n^{\text{target}}$ for $0\leq v_n^{\text{target}}\leq 1$, we conclude the proof.
\end{proof}

\begin{prop}\label{prop2}
	For fixed $\{v_n^{\mathrm{target}}\}_{n=1}^N$, Problem $\mathcal{P}_{4}$ for finding $\{\gamma_n\}_{n=1}^N$ reduces to a convex feasibility test problem as
	\begin{subequations}\label{op5}
		\begin{align}
			\mathcal{P}_{5}: \text {find}\quad &\{\gamma_{n}\}_{n=1}^{N}\\
			\text { s.t. } \quad & \gamma_{n} > 0, \quad \forall n \in[N], \label{p4 gamma}\\
			& \sum_{n=1}^{N} \gamma_{n}\leq1,\\
			\nonumber&\left(v_n+\frac{1}{\psi_{n}^{-1}\left(v_{n}\right)}\right)\gamma_n-\frac{d}{s}c_n\sum_{n^\prime=1}^N\frac{\gamma_{n^\prime}}{c_{n^\prime}}\geq\frac{d}{s\sigma^2},\\&\forall v_n\in(v_n^{\mathrm{target}},1], \forall n \in[N].\label{p4 stf}
		\end{align}
	\end{subequations}
\end{prop}
\begin{proof}\label{prop zeta in obj}
	For fixed slack variables $\{v_n^{\text{target}}\}_{n=1}^N$, the objective function (\ref{opti4 obj}) is irrelevant to the optimization variables $\{\gamma_n\}_{n=1}^N$. With (\ref{phi_n hat}), the constraint (\ref{p4 stf}) is equivalent to (\ref{opti4_psi}). Thus, Problem $\mathcal{P}_4$ reduces to the feasibility test in (\ref{op5}). Since the constraints in (\ref{p4 gamma})-(\ref{p4 stf}) are convex, Problem $\mathcal{P}_5$ is a convex problem.
\end{proof}
\begin{prop}\label{prop4}
	Suppose that for given $\{v_n^{\mathrm{target}}=\tilde{v}_n\}_{n=1}^N$, Problem $\mathcal{P}_{5}$ is feasible. Then for any $\{v_n^{\mathrm{target}}\}_{n=1}^N$ satisfying $v_n^{\text{target}}\geq \tilde{v}_{n}$, for $n=1,\dots,N$, Problem $\mathcal{P}_{5}$ is feasible. On the contrary, suppose that for given $\{v_n^{\mathrm{target}}=\tilde{v}_n\}_{n=1}^N$, Problem $\mathcal{P}_{5}$ is infeasible. Then for any $\{v_n^{\mathrm{target}}\}_{n=1}^N$ satisfying $v_n^{\mathrm{target}}\leq \tilde{v}_{n}$, for $n=1,\dots,N$, Problem $\mathcal{P}_{5}$ is infeasible.
\end{prop}
\begin{proof}
	We define the feasible region of Problem $\mathcal{P}_5$ as $\mathcal{S}(v_1^{\text{target}},\dots,v_N^{\text{target}})$. Based on constraints (\ref{p4 stf}), we see that for any $\{v_n^{\mathrm{target}}\}_{n=1}^N$ satisfying $v_n^{\text{target}}\geq \tilde{v}_{n}$, $\mathcal{S}(v_1^{\text{target}},\dots,v_N^{\text{target}})\subseteq \mathcal{S}(\tilde{v}_1,\dots,\tilde{v}_N)$ holds. Similarly, for any $\{v_n^{\mathrm{target}}\}_{n=1}^N$ satisfying $v_n^{\text{target}}\leq \tilde{v}_{n}$, we obtain $\mathcal{S}(v_1^{\text{target}},\dots,v_N^{\text{target}})\supseteq \mathcal{S}(\tilde{v}_1,\dots,\tilde{v}_N)$. Therefore, Proposition \ref{prop4} holds.
	%	suppose that for given $\{v_n^{\text{target}}=\tilde{v}_n\}_{n=1}^N$, constraint (\ref{p4 stf}) is established. Then for any $\{v_n^{\mathrm{target}}\}_{n=1}^N$ satisfying $v_n^{\text{target}}\geq \tilde{v}_{n}$, for $n=1,\dots,N$, constraint (\ref{p4 stf}) is also established, which further completes the proof of the first part in this proposition. The proof of the second part is proofed similarly, and is omitted for brevity.
\end{proof}

From Propositions \ref{prop1}-\ref{prop2}, to solve Problem $\mathcal{P}_3$, it suffices to find the minimum values of $\{v_n^{\text{target}}\}_{n=1}^N$ such that Problem $\mathcal{P}_5$ is feasible. From Proposition \ref{prop4}, the minimum values of $\{v_n^{\text{target}}\}_{n=1}^N$ can be found by applying bisection search to each $v_n^{\text{target}}$. In addition, it is difficult to handle the constraint (\ref{p4 stf}) of each $v_n$ over the entire continuous region $[v_n^{\text{target}},1)$. In practice, we require that (\ref{p4 stf}) holds only on some discrete points of $v_n$ within $[v_n^{\text{target}},1)$. Then, the feasibility test $\mathcal{P}_5$ can be solved by using standard convex optimization tools. The optimization algorithm is summarized in Algorithm \ref{optimiation algorithm}.

\begin{algorithm}[h]
	\caption{Optimization of OA-FMTL with M-Turbo-CS.}\label{optimiation algorithm}  		
	\begin{algorithmic}[1] %这个1 表示每一行都显示数字 			
		\STATE \textbf{Initialize} $\gamma_n=\frac{1}{N}, \zeta_n=0, v_{n}^\mathrm{target}=1,\forall n\in[N]$
		\STATE Update $\{\zeta\}_{n=1}^N$ via (\ref{zeta_solve1})

		\FOR{$n = 1,2,\dots,N$}
		\STATE Use bisection search to find $v_{n}^\mathrm{target}$ such that Problem $\mathcal{P}_5$ is feasible
		\ENDFOR
		\STATE Update $\{\gamma_n\}_{n=1}^N$ via solving Problem $\mathcal{P}_5$
		\STATE \textbf{Output} $\{\zeta_n,\gamma_n\}_{n=1}^N$

	\end{algorithmic} 	
\end{algorithm}
\section{Numerical Results}\label{sec:num}
\subsection{Experimental Settings}
We validate our proposed OA-FMTL framework with experiments, and provide some baseline schemes for comparison:
\begin{itemize}
	\item Error-free bound: This bound assumes that the ES receives all the local gradients in an error-free fashion and updates the global model by (\ref{update}) and (\ref{update_relation}).
	
	\item OA-FL with TDM and Turbo-CS: Time division multiplexing (TDM) is applied to the tasks, i.e., each task is assigned with an orthogonal time slot to avoid inter-task interference. OA-FL \cite{amiri_machine_2020} is applied in transmission, and Turbo-CS \cite{ma_turbo_2014} is applied to recover the model aggregation at the ES.
	
	\item OA-FMTL with Turbo-CS: The proposed OA-FMTL framework is applied to transmit the model parameters of all the tasks concurrently, and the Turbo-CS algorithm is used to individually recover the model aggregation of each task by treating the inter-task interference as noise.
	
	\item OA-FMTL with M-Turbo-CS: The proposed OA-FMTL framework is applied to transmit the model parameters of all the tasks concurrently, and the M-Turbo-CS algorithm is used to recover the model aggregations of all the tasks, as outlined in Algorithm \ref{sec:system,ssec:alogr,alg:algorithm}.	
\end{itemize} 

To compare the performances of the above schemes, we design two classification experiments using two commonly used databases, namely, MNIST and Human Activity Recognition (HAR). The experimental settings are as  follows.
\begin{itemize}
	\item MNIST experiment: We consider federated image classification tasks on the MNIST dataset of handwritten digits \cite{yann_mnist_1998} and the Fashion-MNIST dataset of fashion clothing \cite{xiao_fashion-mnist_2017} (i.e., $N=2$). For each task, we train a convolutional neural network with two $5\times5$ convolution layers (the first with
	10 channels, the second with 20, each followed with $2\times2$
	max pooling), a fully connected layer with 50 units and
	ReLu activation, and a final softmax output layer (model parameter length $d=10920$). The learning rate is set to $\eta=0.1$, and the compression ratio is set to $2s/d=3/4$.	
	\item HAR experiment: The set of data gathered from accelerometers and gyroscopes of cell phones from 30 individuals performing six different activities including lying-down, standing, walking, sitting, walking-upstairs, and walking-downstairs \cite{anguita_public_2013}. We first divide 30 individuals into two groups (22 cases for model training and the other 8 for model testing); and then divide the training set into four subgroups with each subgroup assigned to an individual task classifying different human activities (i.e., $N=4$). For each task, we train a neural network consisting of one hidden layer with the output size of 100, and the ReLU activation function, and a softmax layer for network output (model parameter length $d=56806$). The learning rate is set to $\eta=0.03$, and the compression ratio is set to $2s/d=4/5$.
\end{itemize}
Other experimental settings are as follows: We consider a system of $M=20$ local devices; the channel gain $h_m^{(t)}$ follows the Rayleigh distribution with unit variance for $\forall m\in[M]$. 

\subsection{Performance Comparisons}
Fig. \ref{fig:turboperformance1} compares the MSEs of the gradient aggregations of various schemes at $t=90$ round in the MNIST experiment. We see that the simulation results agree well with their corresponding state evolutions. We also see that OA-FMTL with M-Turbo-CS performs better in suppressing the inter-task interference than OA-FMTL with Turbo-CS. This is expected since the latter treats the inter-task interference as noise. Besides, due to no inter-task interference in TDM, the OA-FL with TDM and Turbo-CS outperforms the others in terms of MSE. However, we will show that the learning performance of OA-FMTL with M-Turbo-CS is comparable to that of OA-FL with TDM and Turbo-CS, which reveals that OA-FMTL has a strong error-tolerance capability in gradient aggregation. Fig. \ref{fig:learning1} shows the performances of the tasks in terms of test accuracy versus communication round $t$. We observe that the test accuracy of OA-FMTL with M-Turbo-CS is close to that of OA-FL with TDM and Turbo-CS, and is better than that of OA-FMTL with Turbo-CS. We also note that the test accuracy of OA-FMTL with M-Turbo-CS is only about 2\%-5\% lower than that of the ideal error-free bound, which demonstrates the excellent interference suppression capability of OA-FMTL with M-Turbo-CS.
\begin{figure}[h]
	\centering
	\includegraphics[width=0.85\linewidth]{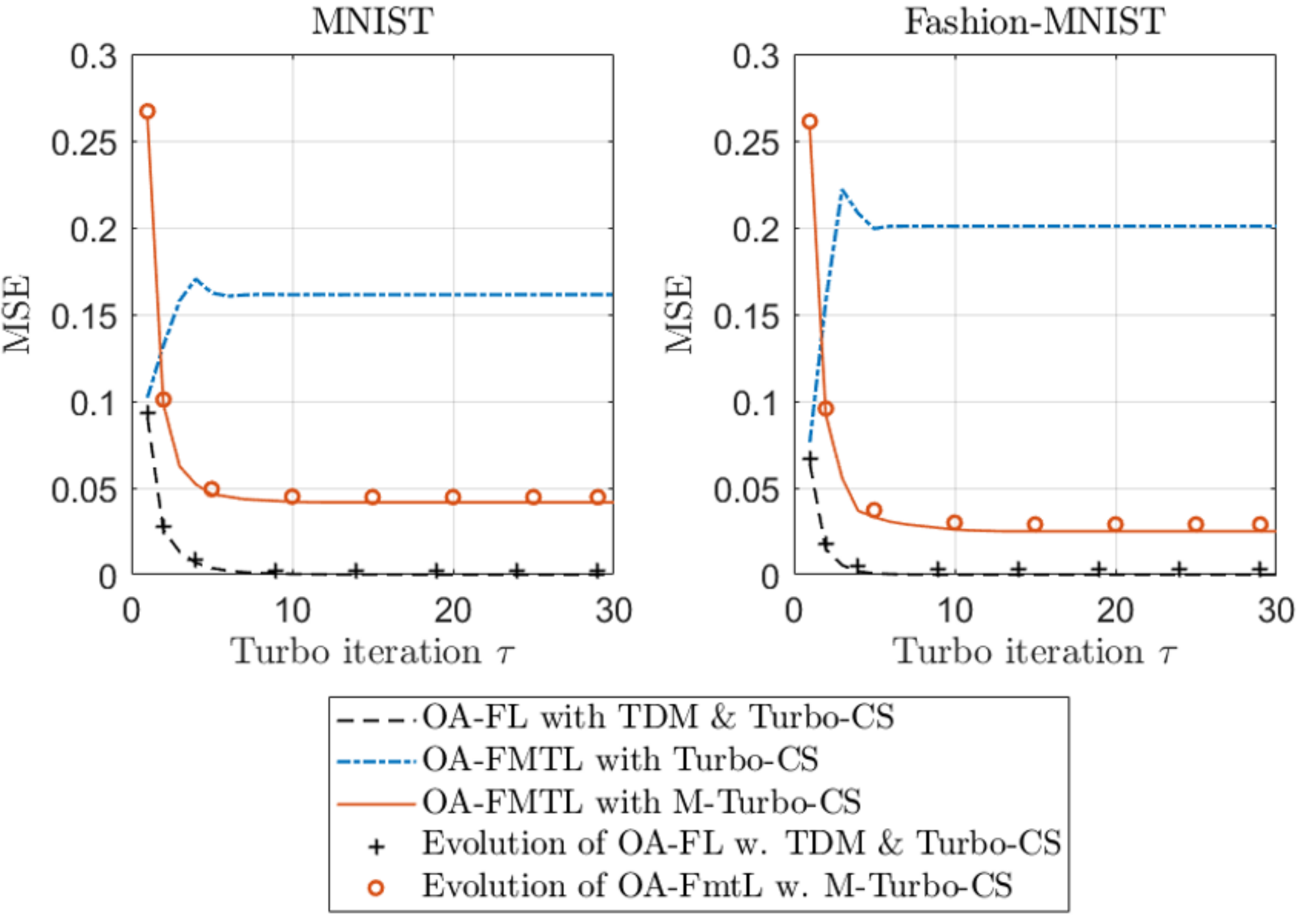}
	\caption{The output MSE on the MNIST experiment at the communication round $t=90$ with SNR $\frac{P_m}{\sigma_w^2}= \SI{20}{dB},\forall m\in[M]$.}
	\label{fig:turboperformance1}
\end{figure}
\begin{figure}[h]
	\centering
	\includegraphics[width=0.9\linewidth]{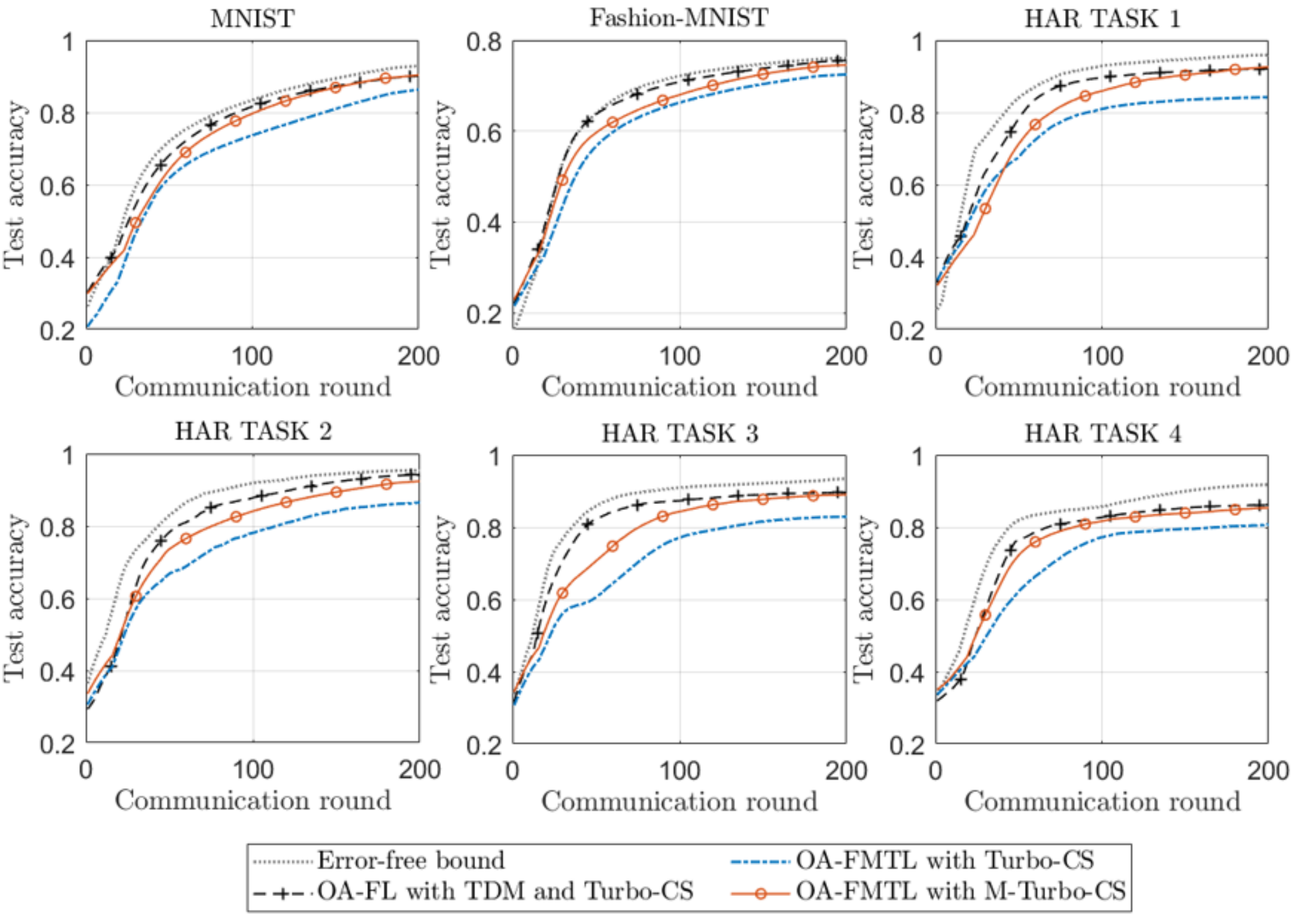}
	\caption{The test accuracies on the two experiments with SNR $\frac{P_m}{\sigma_w^2}= \SI{20}{dB},\forall m\in[M]$.}
	\label{fig:learning1}
\end{figure}

For further comparison, we define $\xi_n^{\text{max}}$ as the maximum test accuracy of each task $n$, and define $t^\star(\xi)$ as the total required rounds of communications for every task $n$ to reach its target accuracy $\xi\xi_n^{\text{max}}$, where $\xi$ is called the relative target accuracy. Thus, $t^\star(\xi)$ of the scheme including the OA-FMTL framework is given by
\begin{equation}\label{sec:system,ssec:perform,equ:MOFL perform 1}		 	 		
	t^\star(\xi) =max\{t_1^\star(\xi),\dots,t_N^\star(\xi)\},
\end{equation}
where $t_n^\star(\xi)$ is the required communication rounds of task $n$ to reach its target accuracy $\xi\xi_n^{\text{max}}$. Besides, $t^\star(\xi)$ of the scheme with TDM is given by
\begin{equation}		 	
	t^\star(\xi) =\sum_{n=0}^N t_n^\star(\xi). 
\end{equation}
Fig. \ref{fig:channel1} depicts the total required communication rounds $t^\star$ versus relative target accuracy $\xi$ on the two experiments. We see that OA-FMTL with M-Turbo-CS significantly outperforms the other two baseline schemes, and that the total required communication rounds of OA-FMTL with M-Turbo-CS to complete $N$ tasks is only $1/N$ of that of OA-FL with TDM and Turbo-CS at any value of $\xi$. In addition, we note that OA-FMTL with Turbo-CS also requires fewer communication rounds than OA-FL with TDM and Turbo-CS, which demonstrates the advantage of non-orthogonal transmission. 
\begin{figure}[h]
	\centering
	\includegraphics[width=0.85\linewidth]{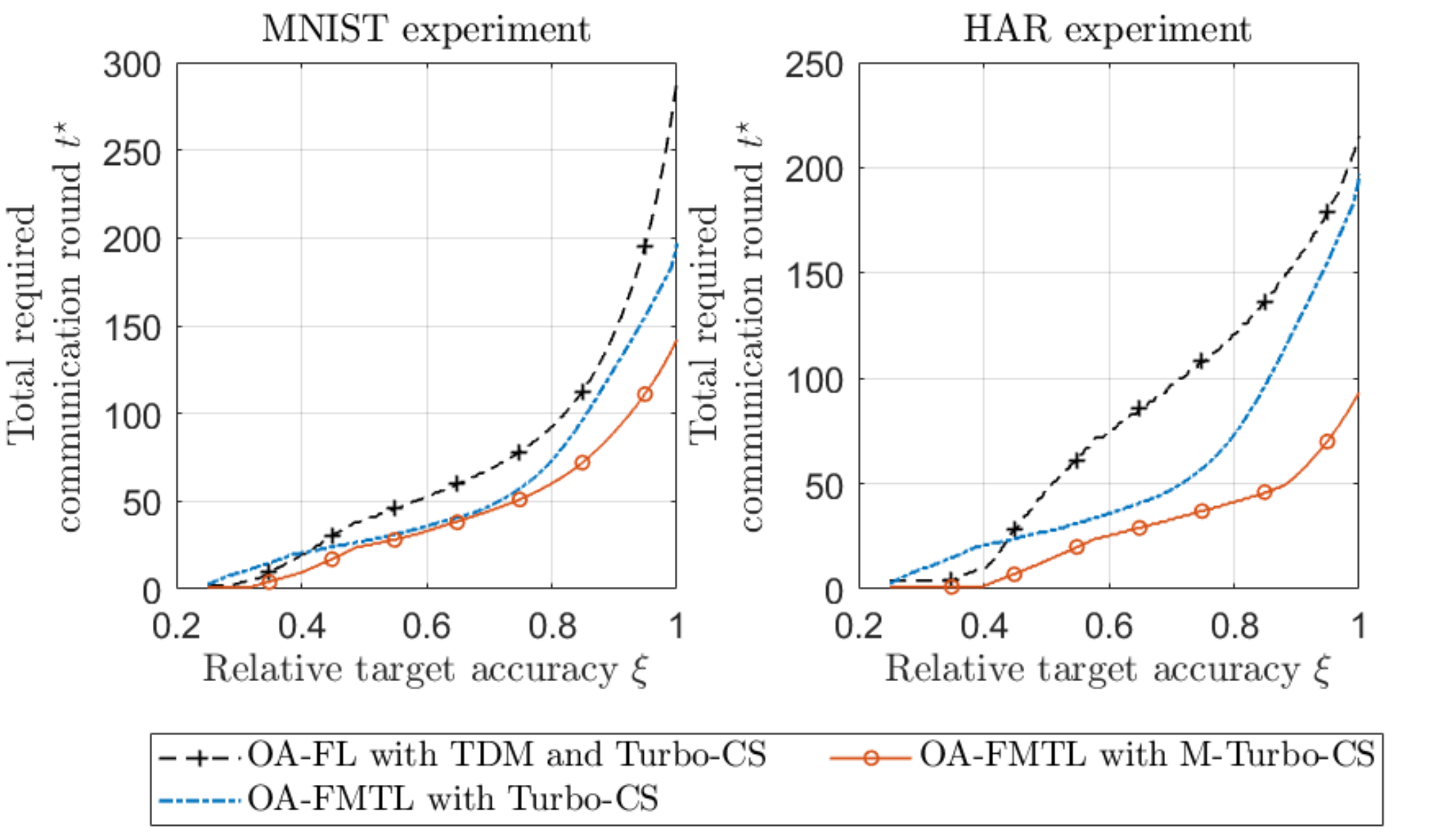}
	\caption{(a) The required communication rounds $t^\star$ of interference-free on the MNIST experiment, with $ \xi^{\text{max}}_1=0.78,\xi^{\text{max}}_2=0.69$ and SNR $\frac{P_m}{\sigma_w^2}= \SI{10}{dB},\forall m\in[M]$. (b) The required communication rounds $t^\star$ of interference-free on the HAR experiment, with $\xi^{\text{max}}_1=0.78,\xi^{\text{max}}_2=0.76,\xi^{\text{max}}_3=0.78,\xi^{\text{max}}_4=0.75$ and SNR $\frac{P_m}{\sigma_w^2}= \SI{20}{dB},\forall m\in[M]$.}
	\label{fig:channel1}
\end{figure}
%\begin{figure}[h]
%	\centering
%	\includegraphics[width=1\linewidth]{turbo_performance2}
%	\caption{The output MSE on MNIST database with $2s/d=3/4$ at the communication round $t=90$. (a) For MNIST task, $\lambda_1^{(90)}=0.5515$. (b) For Fashion-MNIST task, $\lambda_2^{(90)}=0.4728$.}
%	\label{fig:turboperformance2}
%\end{figure}

We next validate the performance of the power allocation optimization for OA-FMTL with M-Turbo-CS. We describe some inter-task power allocation approaches for comparison:
\begin{itemize}
	\item Equal power allocation: This approach assume that an equal power is allocated to each task at each communication round, i.e., we set $\gamma_1=\dots=\gamma_N$ in Problem $\mathcal{P}_1$.
	\item Random power allocation: This approach randomly generates the inter-task power allocation coefficients $\{\gamma_n\}_{n=1}^N$ satisfying $\sum_{n=1}^N\gamma_n=1$.
	\item Optimized power allocation: This approach solves the optimization problem $\mathcal{P}_3$ to obtain the inter-task power allocation coefficients $\{\gamma_n^{(t)}\}_{n=1}^N$ at each communication round $t$.
\end{itemize}
\begin{figure}[!htb]
	\centering
	\includegraphics[width=0.85\linewidth]{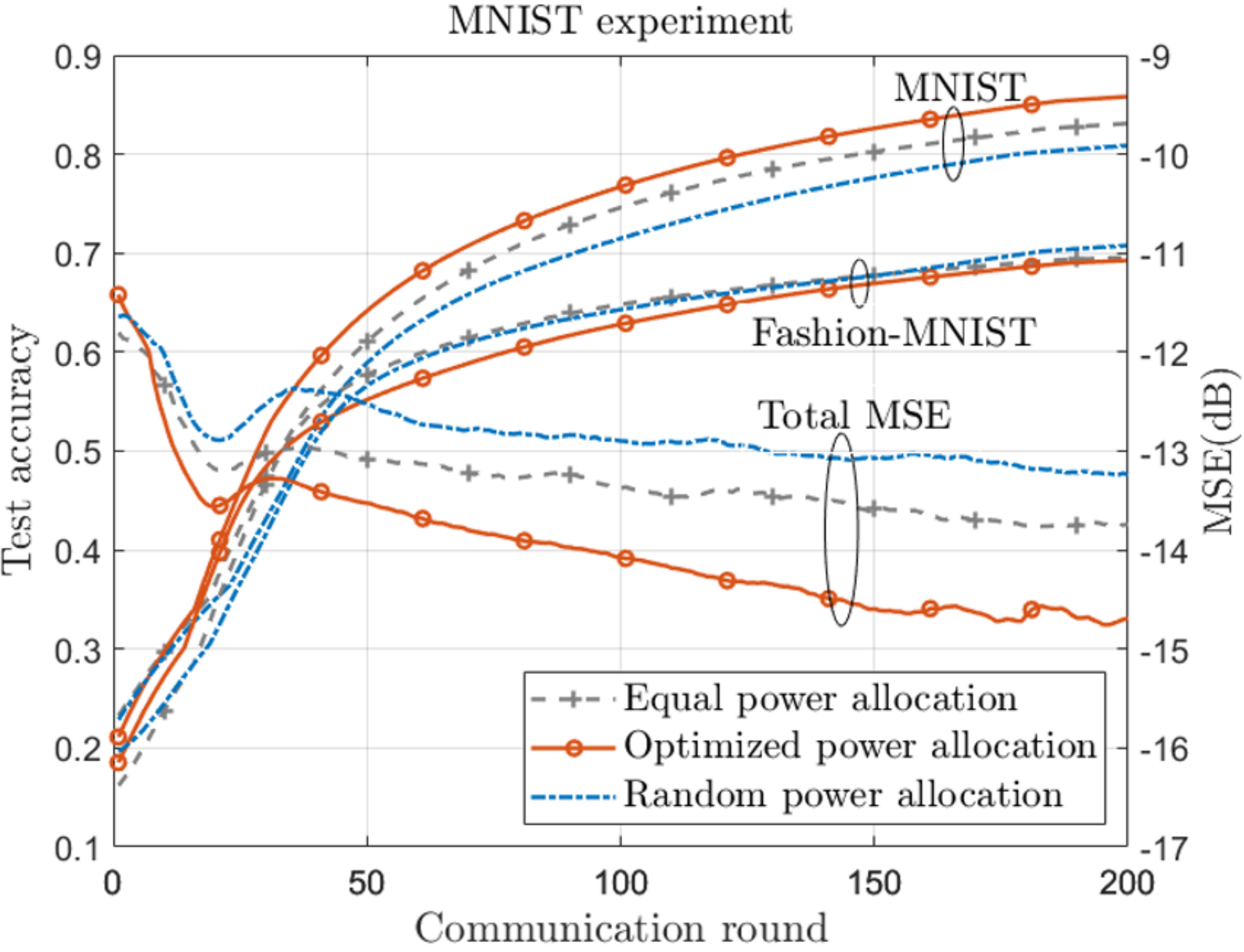}
	\caption{The test accuracies on MNIST experiment with SNR $\frac{P_m}{\sigma_w^2}= \SI{0}{dB},\forall m\in[M]$. Ascending curves: the performance of each task in terms of test accuracy versus communication round $t$. Descending curves: the toatl MSE, i.e., $\sum_{n=1}^N||\bm{e}_n||^2/d$, versus communication round $t$.}
	\label{fig:learning3}
\end{figure}

Fig. \ref{fig:learning3} shows the performance comparison of the three power allocation approaches on the MNIST experiment and the HAR experiment, where these ascending curves measure the performance of all the tasks in terms of test accuracy versus communication round $t$; and these descending curves measure the total MSE, i.e., $\sum_{n=1}^N||\bm{e}_n||^2/d$ defined in (\ref{sec:system,ssec:conve,equ:errorn_all}), versus communication round $t$. Similarly, Fig. \ref{fig:learning2} shows the performance comparison of the three power allocation approaches on the HAR experiment. We see that the accuracies of some tasks are improved, e.g., MNIST, HAR TASK 2 and HAR TASK 3 at the cost of a slight decrease in the accuracies of the other tasks. We also see that optimized power allocation reduces the total MSE as compared with the other approaches. Fig. \ref{fig:channel3} depicts the total required communication rounds $t^\star$ of the three power allocation approaches on the two experiments, versus relative target accuracy $\xi$. We see that random power allocation significantly increases the required number of communication rounds, compared to the other allocation approaches. We also see that our proposed optimized power allocation approach reduces the required communication rounds by $25$ at relative target accuracy $\xi=1$, compared to the equal power allocation approach. This illustrates again that the optimized power allocation approach improves the overall system performance.
\begin{figure}[!htb]
	\centering
	\includegraphics[width=0.85\linewidth]{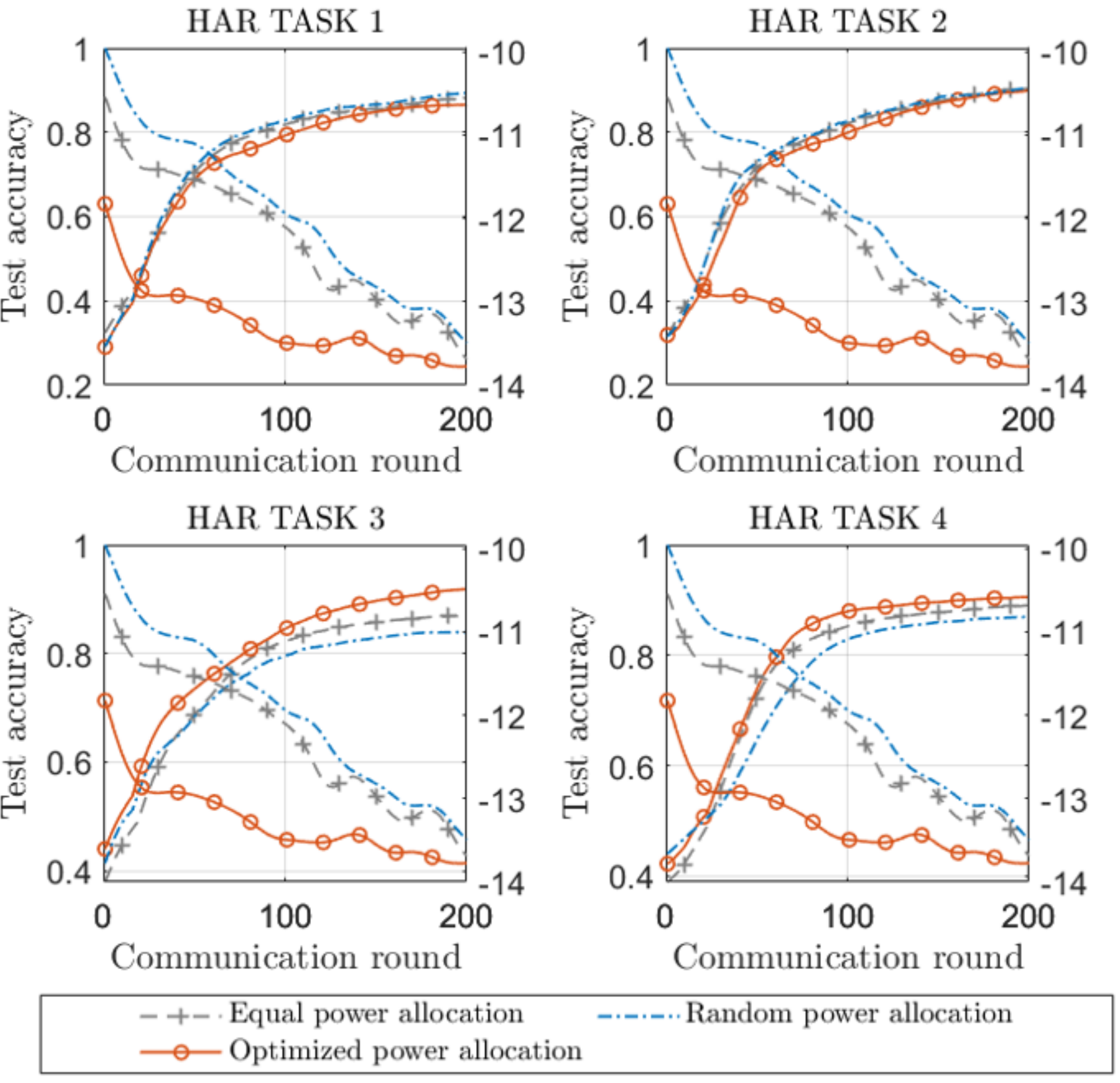}
	\caption{The test accuracies on HAR experiment with SNR $\frac{P_m}{\sigma_w^2}= \SI{0}{dB},\forall m\in[M]$. Ascending curves: the performance of each task in terms of test accuracy versus communication round $t$. Descending curves: the toatl MSE, i.e., $\sum_{n=1}^N||\bm{e}_n||^2/d$, versus communication round $t$.}
	\label{fig:learning2}
\end{figure}
\begin{figure}[!htb]
	\centering
	\includegraphics[width=0.85\linewidth]{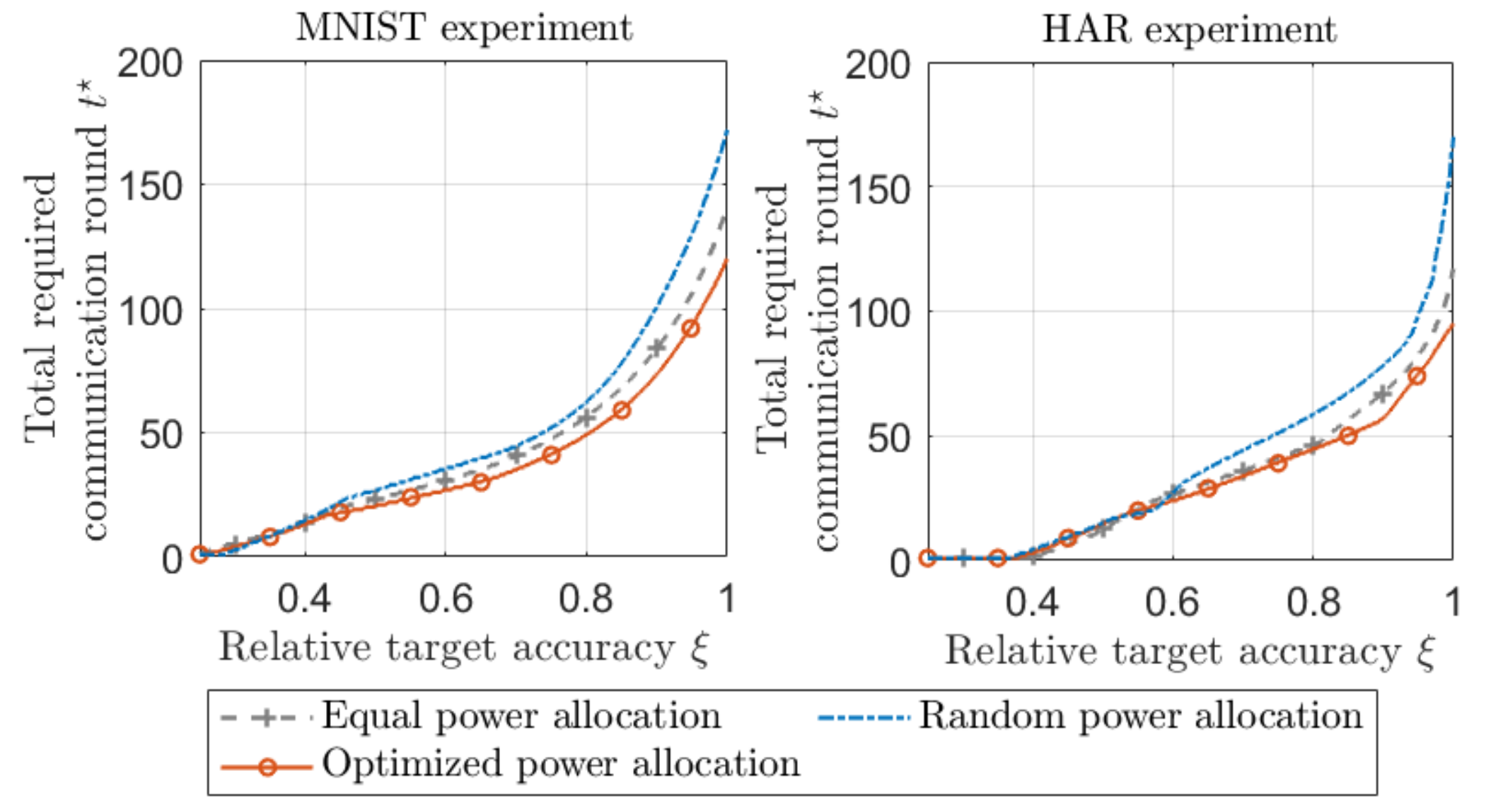}
	\caption{(a) The required communication rounds $t^\star$ of the three power allocation approaches on the MNIST experiment, with $\xi^{\text{max}}_1=0.79,\xi^{\text{max}}_2=0.65$ and the signal-to-noise ratio $\frac{P_m}{\sigma_w^2}= \SI{0}{dB}$. (b) The required communication rounds $t^\star$ of the three power allocation approaches on the HAR experiment, with $ \xi^{\text{max}}_1=0.86,\xi^{\text{max}}_2=0.86,\xi^{\text{max}}_3=0.83,\xi^{\text{max}}_4=0.83$ and the signal-to-noise ratio $\frac{P_m}{\sigma_w^2}= \SI{0}{dB}$.}
	\label{fig:channel3}
\end{figure}

%\begin{figure}
%	\centering
%	\includegraphics[width=0.7\linewidth]{learning4}
%	\caption{The test accuracies on HAR database with $2s/d=4/5$.}
%	\label{fig:learning4}
%\end{figure}

%Meanwhile, we note that our proposed optimization algorithm actually reduces the power of one task to increase the power of another task.  In Fig. \ref{fig:turboperformance2},  

%\begin{figure}[h]
%	\centering
%	\includegraphics[width=0.7\linewidth]{turbo_performance2}
%	\caption{The output MSE on MNIST database with $2s/d=3/4$ at the communication round $t=90$. (a) For MNIST task, $\lambda_1^{(90)}=0.5515$. (b) For Fashion-MNIST task, $\lambda_2^{(90)}=0.4728$.}
%	\label{fig:turboperformance2}
%\end{figure}

\section{Conclusions}
We proposed the over-the-air FMTL (OA-FMTL) framework to achieve communication-efficient FMTL in the presence of inter-task interference. Specifically, we modified the original Turbo-CS algorithm in the compressed sensing context to reconstruct the sparsified model aggregation updates at ES. We further analyzed the performance of the proposed OA-FMTL framework together with the M-Turbo-CS algorithm. Based on that, we formulated a communication-learning optimization problem to improve the system performance by adjusting the power allocation between multiple tasks on the edge devices. We showed that, under mild conditions, the problem can be solved by a feasibility test together with bisection search. Numerical simulations showed that our proposed OA-FMTL can efficiently suppress the inter-task interference, so as to achieve a learning performance comparable to its counterpart with orthogonal multi-task transmission. Meanwhile, simulations also showed that our proposed optimization algorithm further reduces the communication overhead by appropriately adjusting the power allocation among the tasks.

\appendices
\section{Proof of Theorem 1}\label{proof:Theorem 1}
To start with, we bound $||\bm{e}_{n,1}^{(t)}||^2$ as
\begin{align}\label{error:e1}
	\left\|\bm{e}_{n,1}^{(t)}\right\|^2\leq&\left (\frac{2r_n-r_n^{t}-r_n^{t+1}}{1-r_n}\right)^2\\
	\nonumber&\times\left(\beta_{n,1}+\beta_{n,2}\left\|\nabla_n \mathcal{L}_n\left(\bm{\theta}_n^{(t)}\right)\right\|^{2}\right)\\
	\nonumber\leq&\left (\frac{2r_n}{1-r_n}\right)^2\left(\beta_{n,1}+\beta_{n,2}\left\|\nabla_n \mathcal{L}_n\left(\bm{\theta}_n^{(t)}\right)\right\|^{2}\right),
\end{align}
where the first inequality is from \cite[Appendix A]{amiri_machine_2020}, the second inequality is the upper bound obtained by making $t$ tend to infinity, $r_n=\sqrt{(d-k_n)/d}<1$ with $k_n$ defined above (\ref{sec:system,ssec:trans,equ:sparse}), $r_n^t$ denotes the $t$-th power of $r_n$; and  $\beta_{n,1}$, $\beta_{n,2}$ are both constants defined in Assumption \ref{sec:conver,asu:assumption 4}.
Then, by plugging (\ref{normalization}) and the definition of $\bm{g}_n^{(t)}$ in (\ref{sec:system,ssec:trans,equ:channel final}) into (\ref{sec:system,ssec:conve,equ:errorn_all}), we have
\begin{align}
	\mathbb{E}[||\bm{e}_{n,2}^{(t)}||^2]= \mathbb{E}\left[\left|\left|\sum_{m=1}^{M} \bm{g}_{nm}^{\operatorname{no}(t)}(v_{nm}^{(t)}-\zeta_n^{(t)}\sqrt{\gamma_n^{(t)}}h_m^{(t)}\alpha_{m}^{(t)})\right|\right|^2\right],
\end{align}
where $v_{nm}^{(t)}$ is defined in (\ref{normalization}), $\zeta_n^{(t)}$ is the normalization scaling factor of task $n$ in (\ref{es con model updating}), $\gamma_n^{(t)}$ is the inter-task power allocation coefficient of task $n$ in (\ref{sec:system,ssec:trans,equ:mapping}), $\alpha_m^{(t)}$ is the power allocation coefficient of device $m$ in (\ref{sec:system,ssec:trans,equ:power_allocation}), and $h_m^{(t)}$ is the channel coefficient of device $m$ in (\ref{sec:system,ssec:trans,equ:channel}). 
Then,
\begin{align}\label{error:e2}
	\mathbb{E}[||\bm{e}_{n,2}^{(t)}||^2]=\sum_{m=1}^{M} \left(v_{nm}^{(t)}-\zeta_n^{(t)}\sqrt{\gamma_n^{(t)}}h_m^{(t)}\alpha_{m}^{(t)}\right)^2.
\end{align}

Combing (\ref{expectation error}) and (\ref{equ:lemma}), we obtain
\begin{align}\label{error analysis}
	\nonumber\mathbb{E}&[\mathcal{L}_n(\bm{\theta}_n^{(t+1)})]-\mathbb{E}[\mathcal{L}_n(\bm{\theta}_n^{(t)})] \\
	\nonumber&\leq\   	\frac{1}{2L_n}	\left(3\left	(||\bm{e}_{n,1}^{(t)}||^2+\mathbb{E}[||\bm{e}_{n,2}^{(t)}||^2]+\mathbb{E}[||\bm{e}_{n,3}^{(t)}||^2]\right)\right.\\
	&\left.-\mathbb{E}[\|\nabla_n \mathcal{L}_n(\bm{\theta}_n^{(t)})\|^{2}]\right).
\end{align}
Plugging (\ref{turbo error}), (\ref{error:e1}), and (\ref{error:e2}) into (\ref{error analysis}) at the $t$-th training round, we have 		
\begin{align}\label{trainloss_optimumloss}
	\mathbb{E}&[\mathcal{L}_n(\bm{\theta}_n^{(t+1)})]-\mathbb{E}[\mathcal{L}_n(\bm{\theta}_n^{(t)})]\\\nonumber&\leq C_n^{(t)} -\frac{\|\nabla_n \mathcal{L}_n(\bm{\theta}_n^{(t)})\|^{2}}{2L_n}\left(1-3\beta_{n,2}\left(\frac{2r_n}{1-r_n}\right)^2\right),
\end{align}
where $C_n^{(t)}$ is defined in (\ref{define:psi,upsilon_C}). Then, from \cite[eq. (2.4)]{friedlander_hybrid_2012}, we have 
\begin{equation}\label{lemma 2.4}
	\mathbb{E}[||\nabla \mathcal{L}_n(\bm{\theta}_n^{(t)})||^{2}] \geq 2 \Omega_n(\mathbb{E}[\mathcal{L}_n(\bm{\theta}_n^{(t)})]-\mathbb{E}[\mathcal{L}_n(\bm{\theta}_n^{(\star)})]).
\end{equation}
Subtracting $\mathcal{L}_n(\bm{\theta}_n^{(\star)})$ on both sides of (\ref{trainloss_optimumloss}) and plugging (\ref{lemma 2.4}) into (\ref{trainloss_optimumloss}), we obtain
\begin{align}\label{equ:theta_star_one}
	\nonumber\mathbb{E}&[\mathcal{L}_n(\bm{\theta}_n^{(t+1)})]-\mathbb{E}[\mathcal{L}_n(\bm{\theta}_n^{(t)})]\\&\leq
	\left(\mathbb{E}[\mathcal{L}_n(\bm{\theta}_n^{(t)})]-\mathbb{E}[\mathcal{L}_n(\bm{\theta}_n^{(\star)})]\right) \Upsilon_n+ C_n^{(t)},
\end{align}
where $\Upsilon_n$ is defined in (\ref{define:psi,upsilon_C}). Combining (\ref{equ:theta_star_one}) with (\ref{total loss}), we obtain
\begin{align}\label{equ:theta_star_two}
	\nonumber\mathbb{E}&[\mathcal{L}(\bm{\Theta}^{(t+1)})]-\mathbb{E}[\mathcal{L}(\bm{\Theta}^{(t)})]\\
	&\leq
	\left(\mathbb{E}[\mathcal{L}(\bm{\Theta}^{(t)})]-\mathbb{E}[\mathcal{L}(\bm{\Theta}^{(\star)})]\right) \max_n\Upsilon_n+ \sum_{n=1}^NC_n^{(t)}.
\end{align}
Applying the above inequality recursively yields (\ref{equ:final_convergence}), which completes the proof.
% your own custom macro containing the $\backslash${\tt{includegraphics}} command to make things
%\bf{If you include a photo:}\vspace{-33pt}
%\begin{IEEEbiography}[{\includegraphics[width=1in,height=1.25in,clip,keepaspectratio]{fig1}}]{Michael Shell}
%Use $\backslash${\tt{begin\{IEEEbiography\}}} and then for the 1st argument use $\backslash${\tt{includegraphics}} to declare and link the author photo.
%Use the author name as the 3rd argument followed by the biography text.
%\end{IEEEbiography}
\bibliographystyle{IEEEtran}
\bibliography{Output}

\vfill

\end{document}